\documentclass[10pt,a4paper,onecolumn]{IEEEtran}
\usepackage{blindtext, graphicx}
\usepackage{tabularx}
\usepackage{pgf,tikz,pgfplots}
\usetikzlibrary[patterns]
\usepackage[ruled,vlined]{algorithm2e}
\usepackage{hyperref}
\pgfplotsset{compat=1.14}
\usepackage{mathrsfs}
\usetikzlibrary{arrows}
\hypersetup{
 colorlinks,
 linkcolor={blue!100!black},
 citecolor={blue!100!black},
 urlcolor={blue!80!black}
}

\usepackage{amsmath,amssymb,amsfonts,amsthm,graphicx,bm,bbm,dsfont}
\newtheorem{theorem}{Theorem}
\newtheorem{remark}{Remark}

\newtheorem{example}{Example}
\newtheorem{lemma}{Lemma}
\newtheorem{definition}{Definition}

\newtheorem{corollary}{Corollary}

\usepackage{soul}

\DeclareSymbolFont{bbold}{U}{bbold}{m}{n}
\DeclareSymbolFontAlphabet{\mathbbold}{bbold}
\newcommand{\1}{\mathbbold{1}}

\usepackage{color}

\newcommand{\bF}{\mathbb{F}}

\newcommand{\bN}{\mathbb{N}}

\newcommand{\bR}{\mathbb{R}}

\newcommand{\cB}{\mathcal{B}}
\newcommand{\cC}{\mathcal{C}}
\newcommand{\cD}{\mathcal{D}}

\newcommand{\cI}{\mathcal{I}}

\newcommand{\cL}{\mathcal{L}}

\newcommand{\cP}{\mathcal{P}}

\newcommand{\cR}{\mathcal{R}}
\newcommand{\cS}{\mathcal{S}}

\newcommand{\bolda}{\mathbf{a}}
\newcommand{\boldb}{\mathbf{b}}
\newcommand{\boldc}{\mathbf{c}}
\newcommand{\boldd}{\mathbf{d}}
\newcommand{\bolde}{\mathbf{e}}

\newcommand{\bolds}{\mathbf{s}}
\newcommand{\boldt}{\mathbf{t}}

\newcommand{\boldx}{\mathbf{x}}
\newcommand{\boldy}{\mathbf{y}}
\newcommand{\boldz}{\mathbf{z}}

\newcommand{\Space}{{\{0,1\}^L \choose M}}
\newcommand{\SpaceV}[2]{{\{0,1\}^#1 \choose #2}}
\newcommand{\SizeSpace}{{2^L \choose M}}

\newcommand{\Al}{A_{\scriptsize{\mbox{\textit{left}}}}}
\newcommand{\Am}{A_{\scriptsize{\mbox{\textit{mid}}}}}
\newcommand{\Ar}{A_{\scriptsize{\mbox{\textit{right}}}}}

\newcommand{\Apl}{A'_{\scriptsize{\mbox{\textit{left}}}}}
\newcommand{\Apm}{A'_{\scriptsize{\mbox{\textit{mid}}}}}
\newcommand{\Apr}{A'_{\scriptsize{\mbox{\textit{right}}}}}

\makeatletter
\def\namedlabel#1#2{\begingroup
	\def\@currentlabel{#2}%
	\label{#1}\endgroup
}
\makeatother
\begin{document}
\title{On Coding over Sliced Information}

\author{\textbf{Jin Sima}, \IEEEauthorblockN{\textbf{Netanel Raviv}, and \textbf{Jehoshua Bruck}}\\
	\IEEEauthorblockA{
	Department of Electrical Engineering,
California Institute of Technology, Pasadena 91125, CA, USA\\}}

\maketitle

\begin{abstract}
The interest in channel models in which the data is sent as an unordered set of binary strings has increased lately, due to emerging applications in DNA storage, among others. In this paper we analyze the minimal redundancy of binary codes for this channel under substitution errors, and provide several constructions, some of which are shown to be asymptotically optimal up to constants. The surprising result in this paper is that while the information vector is sliced into a set of unordered strings, the amount of redundant bits that are required to correct errors is order-wise equivalent to the amount required in the classical error correcting paradigm.
\end{abstract}

\thispagestyle{empty}



\section{Introduction}\label{section:introduction}
Data storage in synthetic DNA molecules suggests unprecedented advances in density and durability. The interest in DNA storage has increased dramatically in recent years, following a few successful prototype implementations~\cite{Church,Goldman,Erlich,Microsoft}. However, due to biochemical restrictions in synthesis (i.e., writing) and sequencing (i.e., reading), the underlying channel model of DNA storage systems is fundamentally different from its digital-media counterpart. 

Typically, the data in a DNA storage system is stored as a pool of short strings that are dissolved inside a solution, and consequently, these strings are obtained at the decoder in an \textit{unordered} fashion. Furthermore, current technology does not allow the decoder to count the exact number of appearances of each string in the solution, but merely to estimate relative concentrations. These restrictions have re-ignited the interest in \textit{coding over sets}, a model that also finds applications in transmission over asynchronous networks (see Section~\ref{section:PreviousWork}).

In this model, the data to be stored is encoded as a set of~$M$ strings of length~$L$ over a certain alphabet, for some integers~$M$ and~$L$ such that~$M<2^L$; typical values for~$M$ and~$L$ are currently within the order of magnitude of~$10^7$ and~$10^2$, respectively~\cite{Microsoft}. Each individual strings is subject to various types of errors, such as deletions (i.e., omissions of symbols, which result in a shorter string), insertions (which result in a longer string), and substitutions (i.e., replacements of one symbol by another). In the context of DNA storage, after encoding the data as a set of strings over a four-symbol alphabet, the corresponding DNA molecules are synthesized and dissolved inside a solution. Then, a chemical process called \textit{Polymerase Chain Reaction} (PCR) is applied, which drastically amplifies the number of copies of each string. In the reading process, strings whose length is either shorter or longer than~$L$ are discarded, and the remaining ones are clustered according to their respective edit-distance\footnote{The edit distance between two strings is the minimum number of deletions, insertions, and substitutions that turn one to another.}. Then, a majority vote is held within each cluster in order to come up with the most likely origin of the reads in that cluster, and all majority winners are included in the \textit{output set} of the decoding algorithm (Figure~\ref{fig:DNAstoragesystem}).


One of the caveats of this approach is that errors in synthesis might cause the PCR process to amplify a string that was written erroneously, and hence the decoder might include this erroneous string in the output set. In this context, deletions and insertions are easier to handle since they result in a string of length different from\footnote{As long as the number of insertions is not equal to the number of deletions, an event that occurs in negligible probability.}~$L$. Substitution errors, however, are more challenging to combat, and are discussed next.

A substitution error that occurs prior to amplification by PCR can induce either one of two possible error patterns. In one, the newly created string already exists in the set of strings, and hence, the decoder will output a set of~$M-1$ strings. In the other, which is undetectable by counting the size of the output set, the substitution generates a string which is not equal to any other string in the set. In this case the output set has the same size as the error free one. These error patterns, which are referred to simply as \textit{substitutions}, are the main focus of this paper.

Following a formal definition of the channel model in Section~\ref{section:preliminaries}, previous work is discussed in Section~\ref{section:PreviousWork}. Upper and lower bounds on the amount of redundant bits that are required to combat substitutions are given in Section~\ref{section:bounds}. In  Section~\ref{section:OneSubSimple} we provide a construction of a code that can correct a single substitution. This construction is shown to be optimal up to some constant, which is later improved in Appendix~\ref{section:OneSubComplicated}. In Section~\ref{section:MultipleSub} the construction for a single substitution is generalized to multiple substitutions, and is shown to be order-wise optimal whenever the number of substitutions is a constant. 
To further improve the redundancy, we present a sketch of another code construction
in Section~\ref{section:optimal}. The code is capable of correcting with optimal redundancy up to a constant.
Finally, open problems for future research are discussed in Section~\ref{section:FutureWork}.

\begin{figure}
    \centering
    \definecolor{ffwwqq}{rgb}{1,0.4,0}
\definecolor{qqttqq}{rgb}{0,0.2,0}
\definecolor{qqffff}{rgb}{0,1,1}
\definecolor{ffqqqq}{rgb}{1,0,0}
\definecolor{qqccqq}{rgb}{0,0.8,0}
\definecolor{uququq}{rgb}{0.9,0.9,0.9}
\begin{tikzpicture}[line cap=round,line join=round,>=triangle 45,x=0.8cm,y=0.8cm]
\clip(-1,0) rectangle (19,11);

\draw [line width=1pt,color=qqttqq] (-0.5,6.5)-- (3,6.5);
\draw [line width=1pt,color=qqttqq] (3,6.5)-- (3,7.5);
\draw [line width=1pt,color=qqttqq] (3,7.5)-- (-0.5,7.5);
\draw [line width=1pt,color=qqttqq] (-0.5,7.5)-- (-0.5,6.5);

\draw (1,8.8) node[anchor=center] {$\mbox{\textbf{DATA}}$};
\draw (2,8.1) node[anchor=center] {$\Downarrow \mbox{(encoding)}$};
\draw (1.25,7) node[anchor=center] {$\{ \mathbf{x}_i \}_{i=1}^M \in \Space$};
\draw (4.25,7) node[anchor=center] {$\overset{\mbox{(synthesis)}}{\Rightarrow}$};
\draw (12,7) node[anchor=center] {$\overset{\mbox{(PCR)}}{\Rightarrow}$};
\draw (4.25,3) node[anchor=center] {$\overset{ \overset{\mbox{(sequencing)}}{\mbox{(clustering)}}}{\Rightarrow}$};
\draw (12,3) node[anchor=center] {$\overset{\mbox{(majority)}}{\Rightarrow}$};
\draw [line width=3.2pt,color=qqccqq] (16.3,1)-- (15.7,1);
\draw [line width=3.2pt,color=qqccqq] (16.3,3)-- (15.7,3);
\draw [line width=3.2pt,color=qqccqq] (16.3,4)-- (15.7,4);
\draw [line width=3.2pt,color=ffqqqq] (16.3,2)-- (15.7,2);
\draw [line width=3.2pt,color=qqccqq] (7.1,3.4)-- (6.5,3.4);
\draw [line width=3.2pt,color=qqccqq] (7.1,3.7)-- (6.5,3.7);
\draw [line width=3.2pt,color=qqccqq] (7.1,4.0)-- (6.5,4.0);
\draw [line width=3.2pt,color=qqccqq] (7.1,1.1)-- (6.5,1.1);
\draw [line width=3.2pt,color=qqccqq] (7.1,1.4)-- (6.5,1.4);
\draw [line width=3.2pt,color=qqffff] (7.1,1.7)-- (6.5,1.7);
\draw [line width=3.2pt,color=qqccqq] (10,3.4)-- (9.4,3.4);
\draw [line width=3.2pt,color=qqccqq] (10,3.7)-- (9.4,3.7);
\draw [line width=3.2pt,color=qqttqq] (10,4.0)-- (9.4,4.0);
\draw [line width=3.2pt,color=ffwwqq] (10,1.1)-- (9.4,1.1);
\draw [line width=3.2pt,color=ffqqqq] (10,1.4)-- (9.4,1.4);
\draw [line width=3.2pt,color=ffqqqq] (10,1.7)-- (9.4,1.7);
\draw [line width=1pt,dotted] (6.8,3.7) circle (0.8382361703855739cm);
\draw [line width=1pt,dotted] (6.8,1.4) circle (0.8382361703855774cm);
\draw [line width=1pt,dotted] (9.7,3.7) circle (0.8382361703855739cm);
\draw [line width=1pt,dotted] (9.7,1.4) circle (0.8382361703855774cm);
\draw [line width=2pt] (6,8)-- (6,5.3);
\draw [line width=2pt] (6,5.3)-- (10.2,5.3);
\draw [line width=2pt] (10.2,5.3)-- (10.2,8);
\draw [line width=3.2pt,color=qqccqq] (6.954460150010935,5.561782455790734)-- (6.354460150010935,5.561782455790734);
\draw [line width=3.2pt,color=qqccqq] (7.954460150010937,5.561782455790734)-- (7.354460150010934,5.561782455790734);
\draw [line width=3.2pt,color=qqccqq] (8.954460150010936,5.561782455790734)-- (8.354460150010935,5.561782455790734);
\draw [line width=3.2pt,color=ffqqqq] (9.923918136322802,5.555746500962709)-- (9.354460150010935,5.561782455790734);
\draw [line width=2pt] (14,8)-- (14,5.3);
\draw [line width=2pt] (18.2,5.3)-- (18.2,8);
\draw [line width=2pt] (14,5.3)-- (18.2,5.3);
\draw [line width=3.2pt,color=qqccqq] (14.9,5.5)-- (14.3,5.5);
\draw [line width=3.2pt,color=qqccqq] (15.9,5.5)-- (15.3,5.5);
\draw [line width=3.2pt,color=qqccqq] (16.9,5.5)-- (16.3,5.5);
\draw [line width=3.2pt,color=ffqqqq] (17.9,5.5)-- (17.3,5.5);
\draw [line width=3.2pt,color=qqccqq] (14.9,5.8)-- (14.3,5.8);
\draw [line width=3.2pt,color=qqccqq] (15.9,5.8)-- (15.3,5.8);
\draw [line width=3.2pt,color=qqccqq] (16.9,5.8)-- (16.3,5.8);
\draw [line width=3.2pt,color=ffqqqq] (17.9,5.8)-- (17.3,5.8);
\draw [line width=3.2pt,color=qqccqq] (14.9,6.1)-- (14.3,6.1);
\draw [line width=3.2pt,color=qqccqq] (15.9,6.1)-- (15.3,6.1);
\draw [line width=3.2pt,color=qqccqq] (16.9,6.1)-- (16.3,6.1);
\draw [line width=3.2pt,color=ffqqqq] (17.9,6.1)-- (17.3,6.1);
\draw [line width=3.2pt,color=qqccqq] (14.9,6.4)-- (14.3,6.4);
\draw [line width=3.2pt,color=qqccqq] (15.9,6.4)-- (15.3,6.4);
\draw [line width=3.2pt,color=qqccqq] (16.9,6.4)-- (16.3,6.4);
\draw [line width=3.2pt,color=ffqqqq] (17.9,6.4)-- (17.3,6.4);
\draw [line width=3.2pt,color=qqccqq] (14.9,6.7)-- (14.3,6.7);
\draw [line width=3.2pt,color=qqccqq] (15.9,6.7)-- (15.3,6.7);
\draw [line width=3.2pt,color=qqccqq] (16.9,6.7)-- (16.3,6.7);
\draw [line width=3.2pt,color=ffqqqq] (17.9,6.7)-- (17.3,6.7);
\draw [line width=3.2pt,color=qqccqq] (14.9,7.0)-- (14.3,7.0);
\draw [line width=3.2pt,color=qqccqq] (15.9,7.0)-- (15.3,7.0);
\draw [line width=3.2pt,color=qqccqq] (16.9,7.0)-- (16.3,7.0);
\draw [line width=3.2pt,color=ffqqqq] (17.9,7.0)-- (17.3,7.0);
\draw [line width=3.2pt,color=qqccqq] (14.9,7.3)-- (14.3,7.3);
\draw [line width=3.2pt,color=qqccqq] (15.9,7.3)-- (15.3,7.3);
\draw [line width=3.2pt,color=qqccqq] (16.9,7.3)-- (16.3,7.3);
\draw [line width=3.2pt,color=ffqqqq] (17.9,7.3)-- (17.3,7.3);
\draw [line width=3.2pt,color=qqccqq] (14.9,7.6)-- (14.3,7.6);
\draw [line width=3.2pt,color=qqccqq] (15.9,7.6)-- (15.3,7.6);
\draw [line width=3.2pt,color=qqccqq] (16.9,7.6)-- (16.3,7.6);
\draw [line width=3.2pt,color=ffqqqq] (17.9,7.6)-- (17.3,7.6);
\end{tikzpicture}
    \caption{An illustration of a typical operation of a DNA storage system. The data at hand is encoded to a set of~$M$ binary strings of length~$L$ each. These strings are then synthesized, possibly with errors, into DNA sequences, that are placed in a solution and amplified by a PCR process. Then, the DNA sequences are read, clustered by similarity, and the output set is decided by a majority vote. In the illustrated example, one string is synthesized in error, which causes the output set to be in error. If the erroneous string happens to be equal to another existing string, the output set is of size~$M-1$, and otherwise, it is of size~$M$.}
    \label{fig:DNAstoragesystem}
\end{figure}
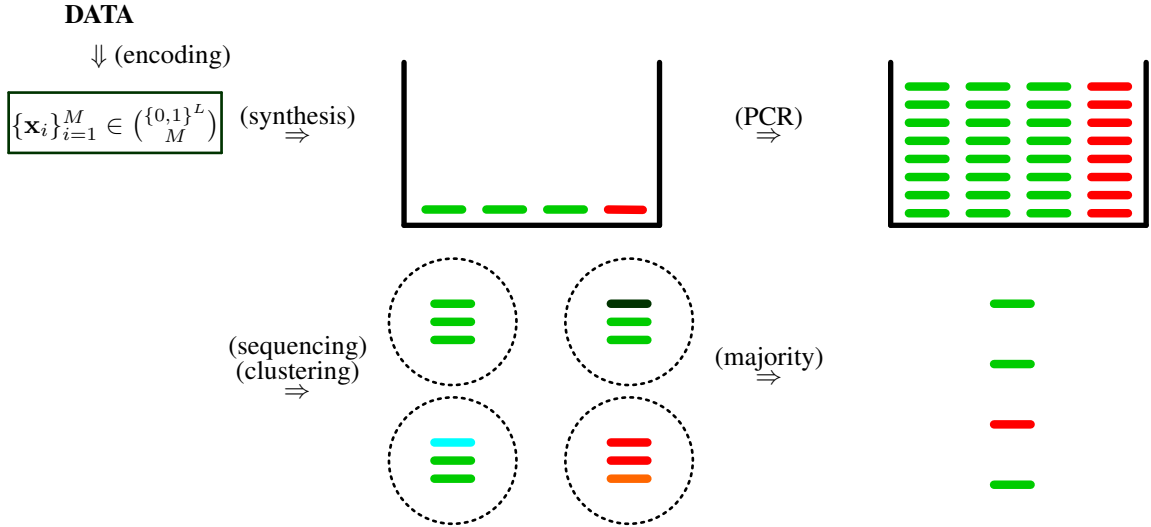

\begin{remark}
    The channel which is discussed in this paper can essentially be seen as taking a string of a certain length~$N$ as input. Then, during transmission, the string is sliced into substrings of equal length, and each substring is subject to substitution errors in the usual sense. Moreover, the order between the slices is lost during transmission, and they arrive as an unordered set.
    
    It follows from the sphere-packing bound~\cite[Sec.~4.2]{Ronny} that \textit{without} the slicing operation, one must introduce at least~$K\log(N)$ redundant bits at the encoder in order to combat~$K$ substitutions. The surprising result of this paper, is that the slicing operation \emph{does not} incur a substantial increase in the amount of redundant bits that are required to correct these~$K$ substitutions. In the case of a single substitution, our codes attain an amount of redundancy that is asymptotically equivalent to the ordinary (i.e., unsliced) channel, whereas for a larger number of substitutions we come close to that, but prove that a comparable amount of redundancy is achievable.
\end{remark}

\section{Preliminaries}\label{section:preliminaries}
To discuss the problem in its most general form, we restrict our attention to binary strings. For integers~$M$ and~$L$ such that\footnote{We occasionally also assume that~$M\le 2^{cL}$ for some~$0<c<1$. This is in accordance with typical values of~$M$ and~$L$ in contemporary DNA storage prototypes (see Section~\ref{section:introduction}).}~$M\le 2^{L}$ we denote by~$\Space$ the family of all subsets of size~$M$ of~$\{0,1\}^L$, and by~$\SpaceV{L}{\le M}$ the family of subsets of size \textit{at most}~$M$ of~$\{0,1\}^L$. In our channel model, a \textit{word} is an element~$W\in\Space$, and a \textit{code}~$\cC\subseteq \Space$ is a set of words (for clarity, we refer to words in a given code as \textit{codewords}). To prevent ambiguity with classical coding theoretic terms, the elements in a word~$W=\{ \boldx_1,\ldots,\boldx_M \}$ are referred to as \textit{strings}. We emphasize that the indexing in~$W$ is merely a notational convenience, e.g., by the lexicographic order of the strings, and this information is not available at the decoder.

For~$K\le ML$, a \textit{$K$-substitution error} ($K$-substitution, in short), is an operation that changes the values of at most~$K$ different positions in a word. Notice that the result of a~$K$-substitution is not necessarily an element of~$\Space$, and might be an element of~${\{0,1\}^L\choose T}$ for some~$M-K\le T\le M$. This gives rise to the following definition.

\begin{definition}\label{definition:ball}
    For a word~$W\in\Space$, a ball~$\cB_K(W)\subseteq \bigcup_{j={M-K}}^M{ \{0,1\}^L\choose j }$ centered at~$W$ is the collection of all subsets of~$\{0,1\}^L$ that can be obtained by a~$K$-substitution in~$W$.
\end{definition}

\begin{example}
    For~$M=2$, $L=3$, $K=1$, and~$W=\{ 001, 011 \}$, we have that
    \begin{align*}
        \cB_K(W)&=\{ \{001,011\},\{ 101, 011 \} , \{011\}, \{ 000,011  \} ,\{ 001,111  \}, \{ 001 \},\{ 001,010 \} \}.
    \end{align*}
\end{example}

In this paper, we discuss bounds and constructions of codes in~$\Space$ that can correct~$K$ substitutions ($K$-substitution codes, for short), for various values of~$K$. The \textit{size} of a code, which is denoted by~$|\cC|$, is the number of codewords (that is, sets) in it. The \textit{redundancy} of the code, a quantity that measures the amount of redundant information that is to be added to the data to guarantee successful decoding, is defined as~$r(\cC)\triangleq \log{2^L\choose M}-\log(|\cC|)$, where the logarithms are in base~$2$.

A code~$\cC$ is used in our channel as follows. First, the data to be stored (or transmitted) is mapped by a bijective \textit{encoding function} to a codeword~$C\in\cC$. This codeword passes through a channel that might introduce up to~$K$ substitutions, and as a result a word~$W\in \cB_K(C)$ is obtained at the decoder. In turn, the decoder applies some \textit{decoding function} to extract the original data. The code~$\cC$ is called a $K$-substitution code if the decoding process always recovers the original data successfully. Having settled the channel model, we are now in a position to formally state our contribution.

\begin{theorem} (Main) For any integers~$M$,~$L$, and~$K$ such that~$M\le 2^{L/(4K+2)}$, there exists an explicit code construction with redundancy~$O(K^2\log(ML))$ (Section~\ref{section:MultipleSub}). For~$K=1$, the redundancy of this construction is at most six times larger than the optimal one (Section~\ref{section:OneSubSimple}). Furthermore, an improved construction for~$K=1$  achieves redundancy which is at most three times the optimal one (Appendix~\ref{section:OneSubComplicated}). 
\end{theorem}
In addition, we sketch an additional construction which achieves optimal redundancy for small (but non-constant) values of~$K$. The full proof will appear in future versions of this paper.

\begin{theorem}
     For integers~$M$,~$L$, and~$K$ that satisfy~$L'+4KL'+2K\log (4KL')\le L$, where~$L'=3\log M  + 4K^2+1$,  there exists an explicit code construction with redundancy $ 2K\log ML + (12K+2)\log M+O(K^3)+O(K\log\log ML)$ (Section~\ref{section:optimal}). The redundancy is at most~$14$ times the optimal one.   
 \end{theorem}

A few auxiliary notions are used throughout the paper, and are introduced herein. For two strings~$\bolds,\boldt\in\{0,1\}^L$, the \textit{Hamming distance}~$d_H(\bolds,\boldt)$ is the number of entries in which they differ. To prevent confusion with common terms, a subset of~$\{0,1\}^L$ is called \textit{a vector-code}, and the set~$\cB_D^H(\bolds)$ of all strings within Hamming distance~$D$ or less of a given string~$\bolds$ is called the \textit{Hamming ball} of radius~$D$ centered at~$\bolds$. A \textit{linear} vector code is called an~$[n,k]_q$ code if the strings in it form a subspace of dimension~$k$ in~$\bF_q^n$, where~$\bF_q$ is the finite field with~$q$ elements. 

Several well-known vector-codes are used in the sequel, such as Reed-Solomon codes or Hamming codes. For an integer~$t$, the Hamming code is an~$[2^t-1,2^t-t-1]_2$ code (i.e., there are~$t$ redundant bits in every codeword), and its minimum Hamming distance is~$3$. Reed-Solomon (RS) codes over~$\bF_q$ exist for every length~$n$ and dimension~$k$, as long as~$q\ge n-1$~\cite[Sec.~5]{Ronny}, and require~$n-k$ redundant symbols in~$\bF_q$. Whenever~$q$ is a power of two, RS codes can be made binary by representing each element of~$\bF_q$ as a binary string of length~$\log_2(q)$. In the sequel we use this form of RS code, which requires~$\log(n)(n-k)$ redundant bits.

Finally, our encoding algorithms make use of \textit{combinatorial numbering maps}~\cite{Knuth}, that are functions that map a number to an element in some structured set. Specifically, $F_{com}:\{1,\ldots,\binom{N}{M}\}\rightarrow \{S:S\subset\{1,\ldots,N\},|S|=M\}$ maps a number to a set of distinct elements, and $F_{perm}:\{1,\ldots,N!\}\rightarrow S_N$ maps a number to a permutation in the symmetric group~$S_N$. The function~$F_{com}$ can be computed using a greedy algorithm with complexity~$O(MN \log N)$, and the function~$F_{perm}$ can be computed in a straightforward manner with complexity~$O(N \log N)$. Using~$F_{com}$ and~$F_{perm}$ together, we define a map~$F:\{1,\ldots,\binom{N}{M}M!\}\rightarrow \{S:S\subset\{1,\ldots,N\},|S|=M\}\times S_M$ that maps a number into an unordered set of size~$M$ together with a permutation. Generally, we denote scalars by lower-case letters~$x,y,\ldots$, vectors by bold symbols~$\boldx,\boldy,\ldots$, integers by capital letters~$K,L,\ldots$, and~$[K]\triangleq\{1,2,\ldots,K\}$.

\section{Previous Work}\label{section:PreviousWork}
The idea of manipulating atomic particles for engineering applications dates back to the 1950's, with R.~Feynman's famous citation ``there's plenty of room at the bottom''~\cite{Feynman}. The specific idea of manipulating DNA molecules for data storage as been circulating the scientific community for a few decades, and yet it was not until 2012-2013 where two prototypes have been implemented~\cite{Church,Goldman}. These prototypes have ignited the imagination of practitioners and theoreticians alike, and many works followed suit with various implementations and channel models~\cite{Chang,Gabrys1,Heckel,Kiah,Raviv,Yazdi}.

By and large, all practical implementations to this day follows the aforementioned channel model, in which multiple short strings are stored inside a solution. Normally, deletions and insertions are also taken into account, but substitutions were found to be the most common form of errors~\cite[Fig.~3.b]{Microsoft}, and strings that were subject to insertions and deletions are scarcer, and can be easily discarded.

The channel model in this work has been studied by several authors in the past. The work of~\cite{Heckel} addressed this channel model under the restriction that individual strings are read in an error free manner, and some strings might get lost as a result of random sampling of the DNA pool. In their techniques, the strings in a codeword are appended with an indexing prefix, a solution which already incurs $\Theta(M\log M)$ redundant bits, or~$\log(e)M-o(1)$ redundancy~\cite[Remark~1]{CodingOverSets}, and will be shown to be strictly sub-optimal in our case.

The recent work of~\cite{CodingOverSets} addressed this model under substitutions, deletions, and insertions. When discussing substitutions only,~\cite{CodingOverSets} suggested a code construction for~$K=1$ with~$2L+1$ bits of redundancy. Furthermore, by using a reduction to constant Hamming weight vector-codes, it is shown that there exists a code that can correct~$e$ errors in each one of the~$M$ sequences with redundancy~$Me\log(L+1)$.

The work of~\cite{Kovacevic} addressed a similar model, where \textit{multisets} are received at the decoder, rather than sets. In addition, errors in the stored strings are not seen in a fine-grained manner. That is, any set of errors in an individual string is counted as a single error, regardless of how many substitutions, insertions, or deletions it contains. As a result, the specific structure of~$\{0,1\}^L$ is immaterial, and the problem reduces to decoding \textit{histograms} over an alphabet of a certain size. 

The specialized reader might suggest the use of \textit{fountain codes}, such as the LT~\cite{LT} codes or Raptor~\cite{Raptor} codes. However, we stress that these solutions rely on randomness at much higher redundancy rates, whereas this work aims for a deterministic and rigorous solution at redundancy which is close to optimal.

Finally, we also mention the \textit{permutation channel}~\cite{perm3,perm2,perm1}, which is similar to our setting, and yet it is farther away in spirit than the aforementioned works. In that channel, a vector over a certain alphabet is transmitted, and its symbols are received at the decoder under a certain permutation. If no restriction is applied over the possible permutations, than this channel reduces to \textit{multiset decoding}, as in~\cite{Kovacevic}. This channel is applicable in networks in which different packets are routed along different paths of varying lengths, and are obtained in an unordered and possibly erroneous form at the decoder. Yet, this line of works is less relevant to ours, and to DNA storage in general, since the specific error pattern in each ``symbol'' (which corresponds to a string in~$\{0,1\}^L$ in our case) is not addressed, and perfect knowledge of the number of appearances of each ``symbol'' is assumed.


\section{Bounds}\label{section:bounds}
In this section we use sphere packing arguments in order to establish an existence result of codes with low redundancy, and a lower bound on the redundancy of any~$K$-substitution code. The latter bound demonstrates the asymptotic optimality of the construction in Section~\ref{section:OneSubSimple} for~$K=1$, up to constants, and near-optimality of the code in Section~\ref{section:optimal}. Our techniques rely on upper and lower bounds on the size of the ball~$\cB_K$ (Definition~\ref{definition:ball}), which are given below. However, since our measure for distance is not a metric, extra care is needed when applying sphere-packing arguments. We begin with the existential upper bound in Subsection~\ref{section:existential}, continue to provide a lower bound for~$K=1$ in Subsection~\ref{section:lowerboundk=1}, and extend this bound to larger values of~$K$ in Subsection~\ref{section:lowerboundK>1}.

\subsection{Existential upper bound}\label{section:existential}
In this subsection, let~$K$, $M$, and~$L$ be positive integers such that~$K\le ML$ and~$M\le 2^L$. The subsequent series of lemmas will eventually lead to the following upper bound.
\begin{theorem}\label{theorem:upperBound}
    There exists a $K$-substitution code~$\cC\subseteq \Space$ such that~$r(\cC)\le 2K\log(ML)+3$.
\end{theorem}

We begin with a simple upper bound on the size of the ball~$\cB_K$.
\begin{lemma}\label{lemma:ballUpperBound}
    For every word~$W=\{\boldx_i\}_{i=1}^M\in \Space$ and every positive integer~$K\le ML$, we have that~$|\cB_K(W)|\le \sum_{\ell=0}^K{ML\choose \ell}$.
\end{lemma}

\begin{proof}
Every word in~$\cB_K(W)$ is obtained by flipping the bits in~$\boldx_i$ that are indexed by some~$J_i\subseteq [L]$, for every~$i\in[M]$, where~$\sum_{i=1}^M|J_i|\le K$. Clearly, there are at most~$\sum_{\ell=0}^K{ML\choose \ell}$ ways to choose the index sets~$\{J_i\}_{i=1}^M$.
\end{proof}

For~$W\in \SpaceV{L}{\le M}$ let~$\cR_K(W)$ be the set of all words~$U\in\Space$ such that~$W\in \cB_K(U)$. That is, for a channel output~$W$, the set~$\cR_K(W)$ contains all potential codewords~$U$ whose transmission through the channel can result in~$W$, given that at most~$K$ substitutions occur. Further, for~$W\in\Space$ define the \textit{confusable set} of~$W$ as~$\cD_K(W)\triangleq \cup_{W'\in\cB_K(W)}\cR_K(W')$. It is readily seen that the words in the confusable set~$\cD_K(W)$ of a word~$W$ cannot reside in the same $K$-substitution code as~$W$, and therefore we have the following lemma.

\begin{lemma}\label{lemma:UpperBoundCorrected}
    For every~$K$, $M$, and~$L$ such that~$K\le ML$ and~$M\le 2^L$ there exists a~$K$-substitution code~$\cC$ such that
    \begin{align*}
        |\cC|&\ge \left\lfloor \frac{\SizeSpace}{D} \right\rfloor,\mbox{ where}\\
        D&\triangleq \max_{W\in\Space}|\cD_K(W)|.
    \end{align*}
\end{lemma}
\begin{proof}
    Initialize a list~$\cP=\Space$, and repeat the following process.
    \begin{enumerate}
        \item Choose~$W\in \cP$.
        \item Remove~$\cD_K(W)$ from~$\cP$.
    \end{enumerate}
    Clearly, the resulting code~$\cC$ is of the aforementioned size. It remains to show that~$\cC$ corrects~$K$ substitutions, i.e., that~$\cB_K(C)\cap \cB_K(C')=\varnothing$ for every distinct~$C,C'\in\cC$.
    
    Assume for contradiction that there exist distinct~$C,C'\in\cC$ and~$V\in\SpaceV{L}{\le M}$ such that $V\in\cB_K(C)\cap \cB_K(C')$, and w.l.o.g assume that~$C$ was chosen earlier than~$C'$ in the above process. Since~$V\in \cB_K(C)$, it follows that~$\cR_K(V)\subseteq \cD_K(C)$. In addition, since~$V\in\cB_K(C')$, it follows that~$C'\in \cR_K(V)$. Therefore, a contradiction is obtained, since~$C'$ is in~$\cD_K(C)$, that was removed from the list~$\cP$ when~$C$ was chosen.
\end{proof}

\begin{lemma}\label{lemma:SizeOfR}
    For an nonnegative integer~$T\le K$ and~$W\in \SpaceV{L}{M-T}$ we have that~$|\cR_K(W)|\le 2(2ML)^K$.
\end{lemma}
\begin{proof}
    Denote~$W=\{ \boldy_1,\ldots,\boldy_{M-T} \}$ and let~$U\in\cR_K(W)$. Notice that by the definition of~$\cR_K(W)$, there exists a $K$-substitution operation which turns~$U$ to~$W$. Therefore, every~$\boldy_i$ in~$W$ is a result of a certain nonnegative number of substitutions in one or more strings in~$U$. Hence, we denote by~$\boldz_1^1,\ldots,\boldz_{i_1}^1$ the strings in~$U$ that resulted in~$\boldy_1$ after the $K$-substitution operation, we denote by~$\boldz_1^2,\ldots,\boldz_{i_2}^2$ the strings which resulted in~$\boldy_2$, and so on, up to~$\boldz_1^{M-T},\ldots,\boldz_{i_{M-T}}^{M-T}$, which resulted in~$\boldy_{M-T}$. Therefore, since~$U=\cup_{j=1}^{M-T}\{ \boldz_1^j,\ldots,\boldz_{i_j}^j \}$, it follows that there exists a set~$\cL\subseteq [M]\times [L]$, of size at most~$K$, such that
    \begin{align}\label{equation:union}
        \begin{pmatrix}
        \boldz^1_1\\
        \vdots\\
        \boldz^1_{i_1}\\
        \boldz^2_1\\
        \vdots\\
        \boldz^{M-T-1}_{i_{M-T-1}}\\
        \boldz^{M-T}_{1}\\
        \vdots\\
        \boldz^{M-T}_{i_{M-T}}
        \end{pmatrix}&=
        \begin{pmatrix}
        \boldy_1\\
        \vdots\\
        \boldy_{1}\\
        \boldy_{2}\\
        \vdots\\
        \boldy_{M-T-1}\\
        \boldy_{M-T}\\
        \vdots\\
        \boldy_{M-T}
        \end{pmatrix}^{(\cL)},
    \end{align}
    where~$(\cdot)^{(\cL)}$ is a matrix operator, which corresponds to flipping the bits that are indexed by~$\cL$ in the matrix on which it operates. In what follows, we bound the number of ways to choose~$\cL$, which will consequently provide a bound on~$|\cR_K(W)|$.
    
    First, 
    define $\mathcal{P}=\{p:i_p>1\}$,
    and denote~$P\triangleq |\cP|$. Therefore, since~$\sum_{j=1}^{M-T}i_j=M$, it follows that
    \begin{align}\label{equation:ipSum}
        \sum_{p\in \cP}i_p=\sum_{j=1}^{M-T}i_j-\sum_{j\notin \cP}i_j=M-(M-T-P)=T+P.
    \end{align}
    Second, notice that for every~$p\in \cP$, the set~$\{ \boldz_1^{p},\ldots, \boldz^{p}_{i_p}\}$ contains~$i_p$ different strings. Hence, since after the~$K$-substitution operation they are all equal to~$\boldy_{p}$, it follows that at least~$i_p-1$ of them must undergo at least one substitution. Clearly, there are~$\binom{i_p}{i_p-1}=i_p$ different ways to choose who will these~$i_p-1$ strings be, and additional~$L^{i_p-1}$ different ways to determine the locations of the substitutions, and therefore $i_p\cdot L^{i_p-1}$ ways to choose these~$i_p-1$ substitutions.
    
    Third, notice that
    \begin{align}\label{equation:K-T}
        K-\sum_{p\in \cP}(i_p-1)=K-\sum_{p\in \cP}i_p+P\overset{\eqref{equation:ipSum}}{=}K-T,
    \end{align}
    and hence, there are at most~$K-T$ remaining positions to be chosen to~$\cL$, after choosing the~$i_p-1$ positions for every~$p\in\cP$ as described above.
    
    Now, let~$\cI$ be the set of all tuples~$i_1,\ldots,i_{M-T}$ of positive integers that sum to~$M$ (whose size is~$\binom{M-1}{M-T-1}$ by the famous \textit{stars and bars} theorem). Let~$N:\cI\to\bN$  be a function which maps $(i_1,\ldots,i_{M-T})\in\cI$ to the number of different~$U\in\cR_K(W)$ for which there exist~$\cL\subseteq [M]\times [L]$ of size at most~$K$ such that~\eqref{equation:union} is satisfied. Since this quantity is at most the number of ways to choose a suitable~$\cL$, the above arguments demonstrate that
    \begin{align*}
        N(i_1,\ldots,i_{M-T})\le \binom{ML}{K-T}\prod_{p\in \cP}i_pL^{i_p-1}.
    \end{align*}
    Then, we have
    \begin{align}\label{equation:RKWtop}
        |\cR_K(W)|&\le 
        \sum_{\cI}N(i_1,\ldots,i_{M-T})\le \sum_{\cI}\binom{ML}{K-T}\prod_{p\in\cP}i_{p}L^{i_{p}-1} \nonumber\\
        &\le\sum_{\cI} (ML)^{K-T}L^{\sum_{p}(i_{p}-1)}\prod_{p\in\cP}i_{p}\overset{\eqref{equation:K-T}}{\le}\sum_{\cI} (ML)^{K-T}L^T\prod_{p\in\cP}i_p.
    \end{align}
    Since the geometric mean of positive numbers is always less than the arithmetic one, we have~$\left(\prod_{p\in\cP}i_p\right)^{1/P}\le \tfrac{1}{P}\sum_{p\in\cP}i_p$, and hence,
    \begin{align}        \eqref{equation:RKWtop}&\le\sum_{\cI} (ML)^{K-T}L^T\left(\tfrac{\sum_p i_p}{P}\right)^P\overset{\eqref{equation:ipSum}}{\le} \sum_{\cI}(ML)^{K-T}L^{T}((T+P)/P)^P\nonumber\\
        &\le \binom{M-1}{M-T-1}(ML)^{K-T}L^{T}((T+P)/P)^P \le\binom{M}{T}(ML)^{K-T}L^{T}((T+P)/P)^P\nonumber\\
        &\le (ML)^{K-T}(ML)^T ((T+P)/P)^P\le (ML)^K((T+P)/P)^P\nonumber\\\label{equation:rk}
        &\overset{(a)}{\le} (ML)^K2^T\le (2ML)^K,
    \end{align}
    where~$(a)$ will be proved in  Appendix~\ref{section:monotonicity}. 
\end{proof}

\begin{proof}(of Theorem~\ref{theorem:upperBound})
    It follows from Lemma~\ref{lemma:ballUpperBound}, Lemma~\ref{lemma:SizeOfR}, and from the definition of~$D$ that
    \begin{align*}
        D \le \max_{W\in\Space}|\cB_K(W)|\cdot \max_{W\in \Space}|\cR_K(W)|\le \left( \sum_{\ell=0}^K\binom{ML}{\ell} \right)(2ML)^K.
    \end{align*}
    Therefore, the code~$\cC$ that is constructed in Lemma~\ref{lemma:UpperBoundCorrected} satisfies
    \begin{align*}
        r(\cC)&\le \log\SizeSpace-\log|\cC|\le \log\left( \left( \sum_{\ell=0}^K\binom{ML}{\ell} \right)(2ML)^K \right)\\
        &=\log\left( \sum_{\ell=0}^K\binom{ML}{\ell} \right)+K(\log(ML)+1)\\
        &\le \log\left( K\binom{ML}{K} \right)+K(\log(ML)+1)\\
        &\le \left( \log(K)-\log(K!)+\log(ML^K) \right)+K(\log(ML)+1)\le 2K\log(ML)+2.\qedhere
    \end{align*}
\end{proof}

\subsection{Lower bound for a single substitution code}\label{section:lowerboundk=1}

Notice that the bound in Lemma~\ref{lemma:ballUpperBound} is tight, e.g., in cases where~$d_H(\boldx_i,\boldx_j)\ge 2K+1$ for all distinct~$i,j\in[M]$. This might occur only if~$M$ is less than the maximum size of a~$K$-substitution correcting vector-code, i.e., when~$M\le 2^L/(\sum_{i=0}^K{L\choose i})$ \cite[Sec.~4.2]{Ronny}. When the minimum Hamming distance between the strings in a codeword is not large enough, different substitution errors might result in identical words, and the size of the ball is smaller than the given upper bound.

\begin{example}
    For~$L=4$ and~$M=2$, consider the word~$W=\{ \underline{0}110, \boldsymbol{0}11\underline{\boldsymbol{1}} \}$. By flipping either the two underlined symbols, or the two bold symbols, the word~$W'=\{ 0110, 1110 \}$ is obtained. Hence, different substitution operation might result in identical words.
\end{example}

However, in some cases it is possible to bound the size of~$\cB_K$ from below by using tools from Fourier analysis of Boolean functions. In the following it is assumed that~$M\le 2^{(1-\epsilon)L}$ for some~$0<\epsilon<1$, and that~$K=1$. A word~$W\in\Space$ corresponds to a Boolean function~$f_W:\{\pm 1\}^L\to \{\pm 1 \}$ as follows. For~$\boldx\in\{0,1\}^L$ let~$\overline{\boldx}\in\{\pm 1\}^L$ be the vector which is obtained from~$\boldx$ be replacing every~$0$ by~$1$ and every~$1$ by~$-1$. Then, we define $f_W(\overline{\boldx})=-1$ if~$\boldx\in W$, and~$1$ otherwise. Considering the set~$\{ \pm 1 \}^L$ as the \textit{hypercube graph}\footnote{The nodes of the hypercube graph of dimension~$L$ are identified by~$\{\pm 1 \}^L$, and every two nodes are connected if and only if the Hamming distance between them is~$1$.}, the \textit{boundary} $\partial f_W$ of~$f_W$ is the set of all edges~$\{\boldx_1,\boldx_2\}\in {\{\pm 1\}^L\choose 2}$ in this graph such that~$f_W(\boldx_1)\ne f_W(\boldx_2)$.

\begin{lemma}\label{lemma:boundary}
    For every word~$W\in \Space$ we have that~$|\cB_1(W)|\ge|\partial f_W|$.
\end{lemma}

\begin{proof}
    Every edge~$e$ on the boundary of~$f_W$ corresponds to a substitution operation that results in a word~$W_e\in\cB_1(W)\cap \Space$. To show that every edge on the boundary corresponds to a \textit{unique} word in~$\cB_1(W)$, assume for contradiction that~$W_e=W_{e'}$ for two distinct edges~$e=\{\overline{\boldx}_1,\overline{\boldx}_2\}$ and~$e'=\{\overline{\boldy}_1,\overline{\boldy}_2\}$, where~$\boldx_1,\boldy_1\in W$ and~$\boldx_2,\boldy_2\notin W$. Since both~$W_e$ and~$W_{e'}$ contain precisely one element which is not in~$W$, and are missing one element which is in~$W$, it follows that~$\boldx_1=\boldy_1$ and~$\boldx_2=\boldy_2$, a contradiction. Therefore, there exists an injective mapping between the boundary of~$f_W$ and~$\cB_1(W)$, and the claim follows.
\end{proof}
Notice that the bound in Lemma~\ref{lemma:boundary} is tight, e.g., in cases where the minimum Hamming distance between the strings of~$W$ is at least~$2$. This implies the tightness of the bound which is given below in these cases. Having established the connection between~$\cB_1(W)$ and the boundary of~$f_W$, the following Fourier analytic claim will aid in proving a lower bound. Let the \textit{total influence} of~$f_W$ be~$I(f_W)\triangleq \sum_{i=1}^L \Pr_{\boldx}(f_W(\boldx)\ne f_W(\boldx^{\oplus i}))$, where~$\boldx^{\oplus i}$ is obtained from~$\boldx$ by changing the sign of the~$i$-th entry, and~$\boldx$ is chosen uniformly at random.

\begin{lemma}~\label{lemma:O'Donnell}
\cite[Theorem~2.39]{O'Donnell} For every function~$f:\{\pm 1\}^L\to \bR$, we have that~$I(f)\ge2\alpha\log(1/\alpha)$, where~$\alpha=\alpha(f)\triangleq \min\{ \Pr_\boldx(f(\boldx)=1),\Pr_\boldx(f(\boldx)=-1)\}$, and~$\boldx\in\{\pm1\}^L$ is chosen uniformly at random.
\end{lemma}


\begin{lemma}\label{lemma:ballLowerBound}
    For every word~$W\in \Space$ we have that~$|\partial f_W|\ge\epsilon ML$.
\end{lemma}

\begin{proof}
    Since~$M\le 2^{(1-\epsilon)L}$ and~$\alpha=\alpha(f_W)=\min\{ (2^L-M)/2^L,M/2^L \}$, it follows that~$\alpha=M/2^{L}$ whenever~$L>\frac{1}{\epsilon}$, which holds for every non-constant~$L$. In addition, since~$\Pr_{\boldx}(f_W(\overline{\boldx})\ne f_W(\overline{\boldx}^{\oplus i}))$ equals the fraction of dimension~$i$ edges that lie on the boundary of~$f_W$ (\cite[Fact~2.14]{O'Donnell}), Lemma~\ref{lemma:boundary} implies that
    \begin{align*}
        I(f_W)=\frac{|\partial f_W|}{2^{L-1}}.
    \end{align*}
    Therefore, since~$M\le 2^{(1-\epsilon)L}$ and from Lemma~\ref{lemma:O'Donnell} we have that $|\partial f_W|=2^{L-1}I(f_W)\ge 2^L\alpha\log(1/\alpha)=M\log(2^L/M)\ge\epsilon ML$.
\end{proof}

\begin{corollary}
    For integers~$L$ and~$M$ and a constant~$0<\epsilon<1$ such that~$M\le2^{(1-\epsilon)L}$, any~$1$-substitution code~$\cC\subseteq \binom{\{0,1\}^L}{M}$ satisfies that~$r(\cC)\ge\log(ML)-O(1)$.
\end{corollary}
\begin{proof}
    According to Lemma~\ref{lemma:boundary} and Lemma~\ref{lemma:ballLowerBound}, every codeword of every~$\cC$ excludes at least~$\epsilon ML$ other words from belonging to~$\cC$. Hence, we have that $|\cC|\le {2^L \choose M}/\epsilon ML$, and by the definition of redundancy, it follows that
    \begin{align*}
        r(\cC)&=\log {2^L\choose M} - \log(|\cC|)\ge \log(\epsilon ML)=\log(ML)-O(1).\qedhere
    \end{align*}
\end{proof}
\subsection{Lower bound for more than one substitution}\label{section:lowerboundK>1}
Similar techniques to the ones in Subsection~\ref{section:lowerboundk=1} can be used to obtain a lower bound for larger values of~$K$. Specifically, we have the following theorem.

\begin{theorem}\label{theorem:lowerBoundGeneralK}
    For integers~$L$,~$M$,~$K$, and positive constants~$\epsilon,c<1$ such that~$M\le 2^{(1-\epsilon)L}$ and~$K\le c\epsilon\sqrt{M}$, a $K$-substitution code~$\cC\subseteq \Space$ satisfies that~$r(\cC)\ge K(\log(ML)-2\log(K))-O(1)$.
\end{theorem}

To prove this theorem, it is shown that certain special~$K$-subsets of~$\partial f_W$ correspond to words in~$\cB_K(W)$, and by bounding the number of these special subsets from below, the lower bound is attained. A subset of~$K$ boundary edges is called \textit{special}, if it does not contain two edges that intersect on a node (i.e., a string) in~$W$. Formally, a subset $\cS\subseteq \partial f_W$ is special if~$|\cS|=K$, and for every~$\{\boldx_1,\boldy_1\},\{ \boldx_2,\boldy_2 \}\in\cS$ with $f_W(\overline{\boldx}_1)=f_W(\overline{\boldx}_2)=-1$ and $f_W(\overline{\boldy}_1)=f_W(\overline{\boldy}_2)=1$ we have that~$\boldx_1\ne\boldx_2$. We begin by showing how special sets are relevant to proving Theorem~\ref{theorem:lowerBoundGeneralK}.

\begin{lemma}\label{lemma:BKballLowerBound}
    For every word~$W\in\Space$ we have that $|\cB_K(W)|\ge \frac{|\{ \cS\subseteq \partial f_W\vert \cS\mbox{ is special} \}|}{K^K}$.
\end{lemma}
\begin{proof}
    It is shown that every special set corresponds to a word in~$\cB_K(W)$, and at most~$K^K$ different special sets can correspond to the same word (namely, there exists a mapping from the family of special sets to~$\cB_K(W)$, which is at most~$K^K$ to~$1$). Let~$\cS=\{ \{ \boldx_i,\boldy_i \} \}_{i=1}^K$ be special, where~$f_W(\overline{\boldx}_i)=-1$ and~$f_W(\overline{\boldy}_i)=1$ for every~$i\in[K]$. Let~$W_\cS\in \binom{\{0,1\}^L}{\le M}$ be obtained from~$W$ by removing the~$\boldx_i$'s and adding the~$\boldy_i$'s, i.e., $W_\cS\triangleq (W\setminus \{ \boldx_i \}_{i=1}^K)\cup\{ \boldy_i \}_{i=1}^T$ for some~$T\le K$; notice that there are exactly~$K$ distinct~$\boldx_i$'s but at most~$K$ distinct~$\boldy_i$'s, since~$\cS$ is special, and therefore we assume w.l.o.g that~$\boldy_1,\ldots,\boldy_T$ are the distinct~$\boldy_i$'s.
    It is readily verified that~$W_\cS\in\cB_K(W)$, since~$W_\cS$ can be obtained from~$W$ by performing $K$ substitution operations in~$W$, each of which corresponds to an edge in~$\cS$. Moreover, every~$\cS$ corresponds to a unique surjective function~$f_\cS:[K]\to[T]$ such that~$f_\cS(i)=j$ if there exists~$j\le T$ such that~$\{\boldx_i,\boldy_j\}\in\cS$, and hence at most~$K^T\le K^K$ different special sets~$\cS$ can correspond to the same word in~$\cB_K(W)$.
\end{proof}

We now turn to prove a lower bound on the number of special sets.
\begin{lemma}\label{lemma:constantc}
    If there exists a positive constant~$c<1$ such that~$K\le c\cdot \epsilon \sqrt{M}$, then there are at least~$(1-c^2)\binom{|\partial f_W|}{K}$ special sets~$\cS\subseteq \partial f_W$.
\end{lemma}
\begin{proof}
    Clearly, the number of ways to choose a $K$-subset of~$\partial f_W$ which is \textit{not} special, i.e., contains~$K$ distinct edges of~$\partial f_W$ but at least two of those are adjacent to the same~$\boldx\in W$, is at most
    \begin{align*}
        M\cdot \binom{L}{2}\cdot \binom{|\partial f_W|}{K-2}=M\cdot\binom{L}{2}\cdot \frac{K(K-1)}{(|\partial f_W|-K+2)(|\partial f_W|-K+1)}\cdot\binom{|\partial f_W|}{K}.
    \end{align*}
Observe that the multiplier of~$\binom{|\partial f_W|}{K}$ in the above expression can be bounded as follows.
\begin{align*}
    M\cdot\binom{L}{2}\cdot \frac{K(K-1)}{(|\partial f_W|-K+2)(|\partial f_W|-K+1)}
    &\le
    M\cdot\binom{L}{2}\cdot \frac{K(K-1)}{(\epsilon ML-K+2)(\epsilon ML-K+1)}
    \\&\le M\cdot L^2\cdot\frac{K^2}{\epsilon^2M^2L^2},
\end{align*}
where the former inequality follows since~$|\partial f_W|\ge \epsilon ML$ by Lemma~\ref{lemma:ballLowerBound}; the latter inequality follows since $K\le c\epsilon \sqrt{M}$ implies that $\epsilon ML-K+2\ge\epsilon ML-K+1\ge \tfrac{1}{\sqrt{2}}\cdot\epsilon ML$ whenever~$\frac{c\sqrt{2}}{\sqrt{2}-1}\le L\sqrt{M}$, which holds for every non-constant~$M$ and~$L$. Therefore, since
\begin{align*}
    M\cdot L^2\cdot\frac{K^2}{\epsilon^2M^2L^2}=\frac{K^2}{\epsilon^2 M}\le c^2,
\end{align*}
it follows that these are at least~$(1-c^2)\cdot\binom{|\partial f_W|}{K}$ special subsets in~$\partial f_W$.
\end{proof}
Lemma~\ref{lemma:BKballLowerBound} and Lemma~\ref{lemma:constantc} readily imply that $|\cB_K(W)|\ge \frac{(1-c^2)}{K^K}\binom{\epsilon ML}{K}$ for every~$W\in\Space$, from which we can prove Theorem~\ref{theorem:lowerBoundGeneralK}.
\begin{proof}
    (of Theorem~\ref{theorem:lowerBoundGeneralK}) Clearly, no two~$K$-balls around codewords in~$\cC$ can intersect, and therefore we must have~$|\cC|\le \binom{2^L}{M}/\min_{W\in\cC}|\cB_K(W)|$. Therefore, 
    \begin{align*}
        r(\cC)&=\log\binom{2^L}{M}-\log|\cC|\ge \log\left( \frac{(1-c^2)}{K^K}\binom{\epsilon ML}{K} \right)\\
        &=\log\binom{\epsilon ML}{K}-K\log(K)-O(1)\\
        &\ge \log\left( \frac{\left(\epsilon ML-K+1 \right)^K}{K^K}
        \right)-K\log(K)-O(1)\\
        &\ge \log\left( \frac{\left(\tfrac{1}{\sqrt{2}}\epsilon ML\right)^K}{K^K}
        \right)-K\log(K)-O(1)\\
        &=K\left( \log(\tfrac{\epsilon}{\sqrt{2}})+\log(ML)) \right)-2K\log(K)-O(1)\\
        &\ge K\log(ML)-2K\log(K)-O(K) \qedhere
    \end{align*}
\end{proof}
\color{black}
\section{Codes for a Single Substitution}\label{section:OneSubSimple}
In this section we present a $1$-substitution code construction that applies whenever~$ M\le 2^{L/6}$, whose redundancy is at most~$3\log ML +3\log M +O(1)$.
For simplicity of illustration, we restrict our attention to values of~$M$ and~$L$ such that~$\log ML +\log M\le M$. 
In the remaining values, a similar construction of comparable redundancy exists. 
\begin{theorem}\label{theorem:1subMain}
	For~$D=\{1,\ldots,\binom{2^{L/3-1}}{M}^3\cdot (M!)^2\cdot 2^{3M-3\log ML-3\log M-6}\}$, there exist an encoding function~$E: D\rightarrow \Space$
	whose image is
	a single substitution correcting code. 
\end{theorem}
The idea behind Theorem~\ref{theorem:1subMain} is to concatenate the strings in a codeword~$C=\{\boldx_i\}_{i=1}^M$ in a certain order, so that
classic $1$-substitution error correction techniques
can be applied over the concatenated string. Since a substitution error may affect any particular order of the~$\boldx_i$'s, we consider the lexicographic orders of several different parts of the~$\boldx_i$'s, instead of the lexicographic order of the whole strings. Specifically, we partition the~$\boldx_i$'s to three parts, and place distinct strings in each of them. Since a substitution operation can scramble the order in at most one part, the correct order will be inferred by a majority vote, so that classic substitution error correction can be applied.

Consider a message~$d\in D$ as a tuple~$d=(d_1,\ldots,d_6)$, where $d_1\in\{1,\ldots,\binom{2^{L/3-1}}{M}\}$, ~$d_3,d_5\in\{1,\ldots,\binom{2^{L/3-1}}{M}M!\}$, and~$d_2,d_4,d_6\in\{1,\ldots,2^{M-\log ML-\log M-2}\}$. Apply the functions~$F_{com}, F_{perm}$, and~$F$ (see Section~\ref{section:preliminaries}) to obtain
\begin{align}\label{equation:combNum}
    F_{com}(d_1) &= \phantom{(}\{\bolda_{1},\ldots,\bolda_{M}\},\nonumber\\
    F(d_3) &= (\{\boldb_{1},\ldots,\boldb_{M}\},\sigma),\nonumber\\
    F(d_5) &= (\{\boldc_{1},\ldots,\boldc_{M}\},\pi),
\end{align}
where~$\bolda_{i},\boldb_i,\boldc_i\in\{0,1\}^{L/3-1}$ for every~$i\in [M]$, the permutations~$\sigma$ and~$\pi$ are in~$S_M$, and the indexing of~$\{\bolda_i\}_{i=1}^M$, $\{\boldb_i\}_{i=1}^M$, and $\{\boldc_i\}_{i=1}^M$ is lexicographic. Further, let~$\boldd_2,\boldd_4$, and~$\boldd_6$ be the binary strings that correspond to~$d_2,d_4$, and~$d_6$, respectively, and let
\begin{alignat}{11}\label{Equation:concatenateSs}
\bolds_1 &= (&&\bolda_{1},&&~\ldots~,&&\bolda_{M},&&\boldb_{\sigma(1)},&&~\ldots~,&&\boldb_{\sigma(M)},&&\boldc_{\pi(1)},&&~\ldots~,&&\boldc_{\pi(M)}&&),	\nonumber\\
\bolds_2 &= (&&\bolda_{\sigma^{-1}(1)},&&~\ldots~,&&\bolda_{\sigma^{-1}(M)},&&\boldb_{1},&&~\ldots~,&&\boldb_{M},&&\boldc_{\sigma^{-1}\pi(1)},&&~\ldots~,&&\boldc_{\sigma^{-1}\pi(M)}&&),\mbox{ and} \nonumber\\
\bolds_3 &= (&&\bolda_{\pi^{-1}(1)},&&~\ldots~,&&\bolda_{\pi^{-1}(M)},&&\boldb_{\pi^{-1}\sigma(1)},&&~\ldots~,&&\boldb_{\pi^{-1}\sigma(M)},&&\boldc_{1},&&~\ldots~,&&\boldc_{M}&&).
\end{alignat}

\begin{center}
\begin{figure}
\centering
\definecolor{ttqqqq}{rgb}{0.2,0,0}
\begin{tikzpicture}[line cap=round,line join=round,>=triangle 45,x=1cm,y=0.8cm]
\clip(-1,0) rectangle (10.2,15);

\fill[line width=2pt,color=ttqqqq,fill=ttqqqq,fill opacity=0.09] (1,11) -- (3,11) -- (3,14) -- (1,14) -- cycle;
\fill[line width=2pt,color=ttqqqq,fill=ttqqqq,fill opacity=0.10000000149011612] (3,14) -- (4,14) -- (4,11) -- (3,11) -- cycle;
\fill[line width=2pt,color=ttqqqq,fill=ttqqqq,fill opacity=0.10000000149011612] (4,14) -- (6,14) -- (6,11) -- (4,11) -- cycle;
\fill[line width=2pt,color=ttqqqq,fill=ttqqqq,pattern=north east lines,pattern color=ttqqqq] (6,14) -- (7,14) -- (7,11) -- (6,11) -- cycle;
\fill[line width=2pt,color=ttqqqq,fill=ttqqqq,fill opacity=0.10000000149011612] (7,14) -- (9,14) -- (9,11) -- (7,11) -- cycle;
\fill[line width=2pt,color=ttqqqq,fill=ttqqqq,pattern=north east lines,pattern color=ttqqqq] (9,14) -- (10,14) -- (10,11) -- (9,11) -- cycle;
\draw [line width=2pt,color=ttqqqq] (1,11)-- (3,11);
\draw [line width=2pt,color=ttqqqq] (3,11)-- (3,14);
\draw [line width=2pt,color=ttqqqq] (3,14)-- (1,14);
\draw [line width=2pt,color=ttqqqq] (1,14)-- (1,11);
\draw [line width=2pt,color=ttqqqq] (3,14)-- (4,14);
\draw [line width=2pt,color=ttqqqq] (4,14)-- (4,11);
\draw [line width=2pt,color=ttqqqq] (4,11)-- (3,11);
\draw [line width=2pt,color=ttqqqq] (3,11)-- (3,14);
\draw [line width=2pt,color=ttqqqq] (4,14)-- (6,14);
\draw [line width=2pt,color=ttqqqq] (6,14)-- (6,11);
\draw [line width=2pt,color=ttqqqq] (6,11)-- (4,11);
\draw [line width=2pt,color=ttqqqq] (4,11)-- (4,14);
\draw [line width=2pt,color=ttqqqq] (6,14)-- (7,14);
\draw [line width=2pt,color=ttqqqq] (7,14)-- (7,11);
\draw [line width=2pt,color=ttqqqq] (7,11)-- (6,11);
\draw [line width=2pt,color=ttqqqq] (6,11)-- (6,14);
\draw [line width=2pt,color=ttqqqq] (7,14)-- (9,14);
\draw [line width=2pt,color=ttqqqq] (9,14)-- (9,11);
\draw [line width=2pt,color=ttqqqq] (9,11)-- (7,11);
\draw [line width=2pt,color=ttqqqq] (7,11)-- (7,14);
\draw [line width=2pt,color=ttqqqq] (9,14)-- (10,14);
\draw [line width=2pt,color=ttqqqq] (10,14)-- (10,11);
\draw [line width=2pt,color=ttqqqq] (10,11)-- (9,11);
\draw [line width=2pt,color=ttqqqq] (9,11)-- (9,14);

\draw [line width=0.5pt,color=ttqqqq] (1,13)-- (3,13);
\draw (2,12.65) node[anchor=center] {$\vdots$};
\draw [line width=0.5pt,color=ttqqqq] (1,12)-- (3,12);

\draw [line width=0.5pt,color=ttqqqq] (4,13)-- (6,13);
\draw (5,12.65) node[anchor=center] {$\vdots$};
\draw [line width=0.5pt,color=ttqqqq] (4,12)-- (6,12);

\draw [line width=0.5pt,color=ttqqqq] (7,13)-- (9,13);
\draw (8,12.65) node[anchor=center] {$\vdots$};
\draw [line width=0.5pt,color=ttqqqq] (7,12)-- (9,12);

\draw (3.5,12.5) node[anchor=center,rotate=-90]{String $\boldz_1$};

\draw (6.5,12.5) node[anchor=center,rotate=-90]{Scrambled bits};

\draw (9.5,12.5) node[anchor=center,rotate=-90]{Scrambled bits};

\draw (0.5,13.5) node[anchor=center]{$\boldx_1$};

\draw (0.5,12.65) node[anchor=center]{$\vdots$};

\draw (0.5,11.5) node[anchor=center]{$\boldx_M$};

\draw (2,13.5) node[anchor=center] {$\textbf{a}_1$};
\draw (2,11.5) node[anchor=center] {$\textbf{a}_M$};
\draw (5,13.5) node[anchor=center] {$\textbf{b}_{\sigma(1)}$};
\draw (5,11.5) node[anchor=center] {$\textbf{b}_{\sigma(M)}$};
\draw (8,13.5) node[anchor=center] {$\textbf{c}_{\pi(1)}$};
\draw (8,11.5) node[anchor=center] {$\textbf{c}_{\pi(M)}$};

\draw (3.5,10.5) node[anchor=center] {$L/3$};
\draw (6.5,10.5) node[anchor=center] {$2L/3$};
\draw (9.5,10.5) node[anchor=center] {$L$};
\draw (5.5,14.5) node[anchor=center] {Ordered by $\mathbf{a}_i$};

\fill[line width=2pt,color=ttqqqq,fill=ttqqqq,fill opacity=0.09] (1,6) -- (3,6) -- (3,9) -- (1,9) -- cycle;
\fill[line width=2pt,color=ttqqqq,fill=ttqqqq,fill opacity=0.10000000149011612] (6,9) -- (7,9) -- (7,6) -- (6,6) -- cycle;
\fill[line width=2pt,color=ttqqqq,fill=ttqqqq,fill opacity=0.10000000149011612] (4,9) -- (6,9) -- (6,6) -- (4,6) -- cycle;
\fill[line width=2pt,color=ttqqqq,fill=ttqqqq,pattern=north east lines,pattern color=ttqqqq] (3,9) -- (4,9) -- (4,6) -- (3,6) -- cycle;
\fill[line width=2pt,color=ttqqqq,fill=ttqqqq,fill opacity=0.10000000149011612] (7,9) -- (9,9) -- (9,6) -- (7,6) -- cycle;
\fill[line width=2pt,color=ttqqqq,fill=ttqqqq,pattern=north east lines,pattern color=ttqqqq] (9,9) -- (10,9) -- (10,6) -- (9,6) -- cycle;
\draw [line width=2pt,color=ttqqqq] (1,6)-- (3,6);
\draw [line width=2pt,color=ttqqqq] (3,6)-- (3,9);
\draw [line width=2pt,color=ttqqqq] (3,9)-- (1,9);
\draw [line width=2pt,color=ttqqqq] (1,9)-- (1,6);
\draw [line width=2pt,color=ttqqqq] (3,9)-- (4,9);
\draw [line width=2pt,color=ttqqqq] (4,9)-- (4,6);
\draw [line width=2pt,color=ttqqqq] (4,6)-- (3,6);
\draw [line width=2pt,color=ttqqqq] (3,6)-- (3,9);
\draw [line width=2pt,color=ttqqqq] (4,9)-- (6,9);
\draw [line width=2pt,color=ttqqqq] (6,9)-- (6,6);
\draw [line width=2pt,color=ttqqqq] (6,6)-- (4,6);
\draw [line width=2pt,color=ttqqqq] (4,6)-- (4,9);
\draw [line width=2pt,color=ttqqqq] (6,9)-- (7,9);
\draw [line width=2pt,color=ttqqqq] (7,9)-- (7,6);
\draw [line width=2pt,color=ttqqqq] (7,6)-- (6,6);
\draw [line width=2pt,color=ttqqqq] (6,6)-- (6,9);
\draw [line width=2pt,color=ttqqqq] (7,9)-- (9,9);
\draw [line width=2pt,color=ttqqqq] (9,9)-- (9,6);
\draw [line width=2pt,color=ttqqqq] (9,6)-- (7,6);
\draw [line width=2pt,color=ttqqqq] (7,6)-- (7,9);
\draw [line width=2pt,color=ttqqqq] (9,9)-- (10,9);
\draw [line width=2pt,color=ttqqqq] (10,9)-- (10,6);
\draw [line width=2pt,color=ttqqqq] (10,6)-- (9,6);
\draw [line width=2pt,color=ttqqqq] (9,6)-- (9,9);

\draw [line width=0.5pt,color=ttqqqq] (1,8)-- (3,8);
\draw (2,7.65) node[anchor=center] {$\vdots$};
\draw [line width=0.5pt,color=ttqqqq] (1,7)-- (3,7);

\draw [line width=0.5pt,color=ttqqqq] (4,8)-- (6,8);
\draw (5,7.65) node[anchor=center] {$\vdots$};
\draw [line width=0.5pt,color=ttqqqq] (4,7)-- (6,7);

\draw [line width=0.5pt,color=ttqqqq] (7,8)-- (9,8);
\draw (8,7.65) node[anchor=center] {$\vdots$};
\draw [line width=0.5pt,color=ttqqqq] (7,7)-- (9,7);

\draw (3.5,7.5) node[anchor=center,rotate=-90]{Scrambled bits};

\draw (6.5,7.5) node[anchor=center,rotate=-90]{String $\boldz_2$};

\draw (9.5,7.5) node[anchor=center,rotate=-90]{Scrambled bits};

\draw (0.1,8.5) node[anchor=center]{$\boldx_{\sigma^{-1}}(1)$};

\draw (0.1,7.65) node[anchor=center]{$\vdots$};

\draw (0.1,6.5) node[anchor=center]{$\boldx_{\sigma^{-1}}(M)$};

\draw (2,8.5) node[anchor=center] {$\textbf{a}_{\sigma^{-1}(1)}$};
\draw (2,6.5) node[anchor=center] {$\textbf{a}_\sigma^{-1}(M)$};
\draw (5,8.5) node[anchor=center] {$\textbf{b}_{1}$};
\draw (5,6.5) node[anchor=center] {$\textbf{b}_{M}$};
\draw (8,8.5) node[anchor=center] {$\textbf{c}_{\sigma^{-1}(\pi(1))}$};
\draw (8,6.5) node[anchor=center] {$\textbf{c}_{\sigma^{-1}(\pi(M))}$};

\draw (3.5,5.5) node[anchor=center] {$L/3$};
\draw (6.5,5.5) node[anchor=center] {$2L/3$};
\draw (9.5,5.5) node[anchor=center] {$L$};
\draw (5.5,9.5) node[anchor=center] {Ordered by $\mathbf{b}_i$};

\fill[line width=2pt,color=ttqqqq,fill=ttqqqq,fill opacity=0.09] (1,1) -- (3,1) -- (3,4) -- (1,4) -- cycle;
\fill[line width=2pt,color=ttqqqq,fill=ttqqqq,pattern=north east lines,pattern color=ttqqqq] (3,4) -- (4,4) -- (4,1) -- (3,1) -- cycle;
\fill[line width=2pt,color=ttqqqq,fill=ttqqqq,fill opacity=0.10000000149011612] (4,4) -- (6,4) -- (6,1) -- (4,1) -- cycle;
\fill[line width=2pt,color=ttqqqq,fill=ttqqqq,pattern=north east lines,pattern color=ttqqqq] (6,4) -- (7,4) -- (7,1) -- (6,1) -- cycle;
\fill[line width=2pt,color=ttqqqq,fill=ttqqqq,fill opacity=0.10000000149011612] (7,4) -- (9,4) -- (9,1) -- (7,1) -- cycle;
\fill[line width=2pt,color=ttqqqq,fill=ttqqqq,fill opacity=0.10000000149011612] (9,4) -- (10,4) -- (10,1) -- (9,1) -- cycle;
\draw [line width=2pt,color=ttqqqq] (1,1)-- (3,1);
\draw [line width=2pt,color=ttqqqq] (3,1)-- (3,4);
\draw [line width=2pt,color=ttqqqq] (3,4)-- (1,4);
\draw [line width=2pt,color=ttqqqq] (1,4)-- (1,1);
\draw [line width=2pt,color=ttqqqq] (3,4)-- (4,4);
\draw [line width=2pt,color=ttqqqq] (4,4)-- (4,1);
\draw [line width=2pt,color=ttqqqq] (4,1)-- (3,1);
\draw [line width=2pt,color=ttqqqq] (3,1)-- (3,4);
\draw [line width=2pt,color=ttqqqq] (4,4)-- (6,4);
\draw [line width=2pt,color=ttqqqq] (6,4)-- (6,1);
\draw [line width=2pt,color=ttqqqq] (6,1)-- (4,1);
\draw [line width=2pt,color=ttqqqq] (4,1)-- (4,4);
\draw [line width=2pt,color=ttqqqq] (6,4)-- (7,4);
\draw [line width=2pt,color=ttqqqq] (7,4)-- (7,1);
\draw [line width=2pt,color=ttqqqq] (7,1)-- (6,1);
\draw [line width=2pt,color=ttqqqq] (6,1)-- (6,4);
\draw [line width=2pt,color=ttqqqq] (7,4)-- (9,4);
\draw [line width=2pt,color=ttqqqq] (9,4)-- (9,1);
\draw [line width=2pt,color=ttqqqq] (9,1)-- (7,1);
\draw [line width=2pt,color=ttqqqq] (7,1)-- (7,4);
\draw [line width=2pt,color=ttqqqq] (9,4)-- (10,4);
\draw [line width=2pt,color=ttqqqq] (10,4)-- (10,1);
\draw [line width=2pt,color=ttqqqq] (10,1)-- (9,1);
\draw [line width=2pt,color=ttqqqq] (9,1)-- (9,4);

\draw [line width=0.5pt,color=ttqqqq] (1,3)-- (3,3);
\draw (2,2.65) node[anchor=center] {$\vdots$};
\draw [line width=0.5pt,color=ttqqqq] (1,2)-- (3,2);

\draw [line width=0.5pt,color=ttqqqq] (4,3)-- (6,3);
\draw (5,2.65) node[anchor=center] {$\vdots$};
\draw [line width=0.5pt,color=ttqqqq] (4,2)-- (6,2);

\draw [line width=0.5pt,color=ttqqqq] (7,3)-- (9,3);
\draw (8,2.65) node[anchor=center] {$\vdots$};
\draw [line width=0.5pt,color=ttqqqq] (7,2)-- (9,2);

\draw (3.5,2.5) node[anchor=center,rotate=-90]{Scrambled bits};

\draw (6.5,2.5) node[anchor=center,rotate=-90]{Scrambled bits};

\draw (9.5,2.5) node[anchor=center,rotate=-90]{String $\boldz_3$};

\draw (0,3.5) node[anchor=center]{$\boldx_{\pi^{-1}(1)}$};

\draw (0,2.65) node[anchor=center]{$\vdots$};

\draw (0,1.5) node[anchor=center]{$\boldx_{\pi^{-1}(M)}$};

\draw (2,3.5) node[anchor=center] {$\textbf{a}_{\pi^{-1}(1)}$};
\draw (2,1.5) node[anchor=center] {$\textbf{a}_{\pi^{-1}(M)}$};
\draw (5,3.5) node[anchor=center] {$\textbf{b}_{\pi^{-1}\sigma(1)}$};
\draw (5,1.5) node[anchor=center] {$\textbf{b}_{\pi^{-1}\sigma(M)}$};
\draw (8,3.5) node[anchor=center] {$\textbf{c}_{1}$};
\draw (8,1.5) node[anchor=center] {$\textbf{c}_{M}$};

\draw (3.5,0.5) node[anchor=center] {$L/3$};
\draw (6.5,0.5) node[anchor=center] {$2L/3$};
\draw (9.5,0.5) node[anchor=center] {$L$};
\draw (5.5,4.5) node[anchor=center] {Ordered by $\mathbf{c}_i$};
\begin{scriptsize}
\end{scriptsize}
\end{tikzpicture}\caption{This figure illustrates the three different~$M\times L$ binary matrices which results from placing the strings~$\{\boldx_i\}_{i=1}^M$ on top of one another in various orders. That is, every row in the above matrices equals to some~$\boldx_i$. Notice that the strings~$\boldz_1$,~$\boldz_2$, and~$\boldz_3$ constitute three~$M\times 1$ columns, that contain the bits of $(\boldd_2,E_H(\boldd_2),E_H(\bolds_1))$, $(\boldd_4,E_H(\boldd_4),E_H(\bolds_2))$, and $(\boldd_6,E_H(\boldd_6),E_H(\bolds_3))$ respectively. For example, when sorting the~$\boldx_i$'s according to the~$\bolda_i$'s (top figure), the bits of $\boldd_2$, $E_H(\boldd_2)$, and~$E_H(\bolds_1)$ appear consecutively.}
\label{figure:Encoding}
\end{figure}
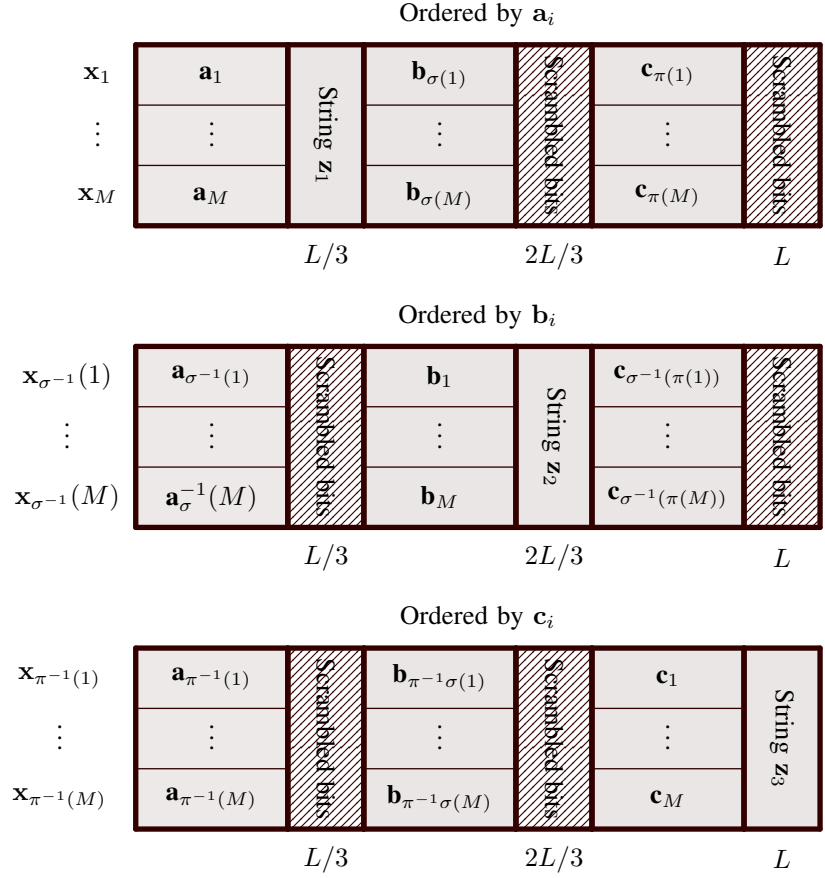
\end{center}

Without loss of generality\footnote{Every string can be padded with zeros to extend its length to~$2^t-t-1$ for some~$t$. It is readily verified that this operation extends the string by at most a factor of two, and by the properties of the Hamming code, this will increase the number of redundant bits by at most~$1$.} assume that there exists an integer~$t$ for which~$|\bolds_i|=(L-3)M=2^t-t-1$ for all~$i\in[3]$. Then, each~$\bolds_i$ can be encoded by using a \textit{systematic}~$[2^t-1,2^t-t-1]_2$ Hamming code, by introducing~$t$ redundant bits. That is, the encoding function is of the form~$\bolds_i\mapsto(\bolds_i,E_H(\bolds_i))$, where~$E_H(\bolds_i)$ are the~$t$ redundant bits, and~$t\le \log(ML)+1$. 
Similarly, we assume that there exists an integer~$h$ for which~$|\boldd_i|=2^h-h-1$ for~$i\in\{2,4,6\}$, and let~$E_H(\boldd_i)$ be the corresponding~$h$ bits of redundancy, that result from encoding~$\boldd_i$ by using a~$[2^h-1,2^h-h-1]$ Hamming code. By the properties of a Hamming code, and by the definition of~$h$, we have that~$h\le \log(M)+1$.

The data~$d\in D$ is mapped to a codeword~$C=\{\boldx_1,\ldots,\boldx_M\}$ as follows, and the reader is encouraged to refer to Figure~\ref{figure:Encoding} for clarifications.
First, we place~$\{\bolda_i\}_{i=1}^M$, $\{\boldb_i\}_{i=1}^M$, and $\{\boldc_i\}_{i=1}^M$ in the different thirds of the~$\boldx_i$'s, sorted by~$\sigma$ and~$\pi$. That is, denoting~$\boldx_i=(x_{i,1},\ldots,x_{i,L})$, we define
\begin{align}\label{equation:DefXsWithabc}
    (x_{i,1},\ldots,x_{i,L/3-1})&=\bolda_i, \nonumber\\
    (x_{i,L/3+1},\ldots,x_{i,2L/3-1})&=\boldb_{\sigma(i)},\mbox{ and}\nonumber\\
    (x_{i,2L/3+1},\ldots,x_{i,L-1})&=\boldc_{\pi(i)}.
\end{align}
The remaining bits~$\{x_{i,L/3} \}_{i=1}^M$, $\{x_{i,2L/3} \}_{i=1}^M$, and~$\{x_{i,L} \}_{i=1}^M$ are used to accommodate the information bits of~$\boldd_2,\boldd_4,\boldd_6$, and the redundancy bits~$\{E_H(\bolds_i)\}^3_{i=1}$ and~$\{E_H(\boldd_i)\}_{i\in\{2,4,6\}}$, in the following manner.
\begin{align}\label{Equation:informationConcise}
x_{i,L/3} &= \begin{cases}
d_{2,i}, &\text{if $i\le M-\log ML-\log M-2$}\\
E_{H}(\boldd_2)_{i-(M-\log ML-\log M-2)}, &\text{if $M-\log ML-\log M -1\le i\le M-\log ML-1$}\\
E_{H}(\bolds_1)_{i-(M-\log ML-1)}, &\text{if $ M-\log ML\le i\le M$,}
\end{cases},	\nonumber\\
x_{i,2L/3} &= \begin{cases}
d_{4,i}, &\text{if $\sigma^{-1}(i)\le M-\log ML-\log M-2$}\\
E_{H}(\boldd_4)_{i-(M-\log ML-\log M-2)}, &\text{if $M-\log ML-\log M -1\le \sigma^{-1}(i)\le M-\log ML-1$}\\
E_{H}(\bolds_2)_{i-(M-\log ML-1)}, &\text{if $ M-\log ML\le \sigma^{-1}(i)\le M$,}
\end{cases},	\nonumber\\
x_{i,L} &= \begin{cases}
d_{6,i}, &\text{if $\pi^{-1}(i)\le M-\log ML-\log M-2$}\\
E_{H}(\boldd_6)_{i-(M-\log ML-\log M-2)}, &\text{if $M-\log ML-\log M -1\le \pi^{-1}(i)\le M-\log ML-1$}\\
E_{H}(\bolds_3)_{i-(M-\log ML-1)}, &\text{if $ M-\log ML\le \pi^{-1}(i)\le M$,}
\end{cases}.
\end{align}
That is, if the strings~$\{\boldx_i\}_{i=1}^M$ are sorted according to the content of the bits~$(x_{i,1},\ldots,x_{i,L/3-1})=\bolda_i$, then the top $M-\log ML\log M-2$ bits of the~$(L/3)$'th column\footnote{Sorting the strings~$\{\boldx_i\}_{i=1}^M$ by any ordering method provides a matrix in a natural way, and can consider columns in this matrix.} contain~$\boldd_2$, the middle~$\log M+1$ bits contain~$E_H(\boldd_2)$, and the bottom~${\log ML+1}$ bits contain~$E_H(\bolds_1)$. Similarly, if the strings are sorted according to $(x_{i,L/3+1},\ldots,x_{i,2L/3-1})=\boldb_i$, then the top~$M-\log ML\log M-2$ bits of the~$(2L/3)$'th column contain~$\boldd_4$, the middle~${\log M+1}$ bits contain~$E_H(\boldd_4)$, and the bottom~$\log ML+1$ bits contain~$E_H(\bolds_2)$, and so on. This concludes the encoding function~$E$ of Theorem~\ref{theorem:1subMain}. It can be readily verified that~$E$ is injective since different messages result in either different~$(\{\bolda_i\}^M_{i=1}$,$\{\boldb_i\}^M_{i=1}$,$\{\boldc_i\}^M_{i=1})$ or the same $(\{\bolda_i\}^M_{i=1}$,$\{\boldb_i\}^M_{i=1}$,$\{\boldc_i\}^M_{i=1})$ with different ~$(\boldd_2,\boldd_4,\boldd_6)$. In either case, the resulting codewords~$\{\boldx_i\}^M_{i=1}$ of the two messages are different.

To verify that the image of~$E$ is a $1$-substitution code, observe first that since~$\{\bolda_i\}_{i=1}^M$, $\{\boldb_i\}_{i=1}^M$, and $\{\boldc_i\}_{i=1}^M$ are \textit{sets}, it follows that any two strings in the same set are distinct. Hence, according to~\eqref{equation:DefXsWithabc}, it follows that~$d_H(\boldx_i,\boldx_j)\ge 3$ for every distinct~$i$ and~$j$ in~$[M]$. Therefore, no $1$-substitution error can cause one~$\boldx_i$ to be equal to another, and consequently, the result of a $1$-substitution error is always in~$\Space$. In what follows a decoding algorithm is presented, whose input is a codeword that was distorted by at most a single substitution, and its output is~$d$.

Upon receiving a word~$C'=\{\boldx'_1,\ldots,\boldx'_M\}\in \cB_1(C)$ for some codeword~$C$ (once again, the indexing of the elements of~$C'$ is lexicographic), we define
\begin{align}\label{equation:abcHat}
    \hat{\bolda}_i&=(x'_{i,1},\ldots,x'_{i,L/3-1})\nonumber\\
    \hat{\boldb}_i&=(x'_{\tau^{-1}(i),L/3+1},\ldots,x'_{\tau^{-1}(i),2L/3-1})\\
    \hat{\boldc}_i&=(x'_{\rho^{-1}(i),2L/3+1},\ldots,x'_{\rho^{-1}(i),L-1})\nonumber,
\end{align}
where
$\tau$ is the permutation by which ~$\{\boldx'_i\}_{i=1}^M$ are sorted according to their~$L/3+1,\ldots,2L/3-1$ entries, and~$\rho$ is the permutation by which they are sorted according to their~$2L/3+1,\ldots,L-1$ entries (we emphasize that~$\tau$ and~$\rho$ are unrelated to the original~$\pi$ and~$\sigma$, and those will be decoded later). Further, when ordering~$\{\boldx_i'\}_{i=1}^M$ by either the lexicographic ordering, by~$\tau$, or by~$\rho$, we obtain \textit{candidates} for each one of~$\boldd_2$, $\boldd_4$, $\boldd_6$, $E_H(\boldd_2)$, $E_H(\boldd_4)$, $E_H(\boldd_6)$, $E_H(\bolds_1)$, $E_H(\bolds_2)$, and $E_H(\bolds_3)$, that we similarly denote with an additional apostrophe\footnote{That is, each one of~$\boldd_2'$, $\boldd_4'$, etc., is obtained from~$\boldd_2$, $\boldd_4$, etc., by at most a single substitution.}. For example, if we order~$\{\boldx_i'\}_{i=1}^M$ according to~$\tau$, then the bottom~$\log(ML)+1$ bits of the~$(2L/3)$-th column are~$E_H(\bolds_2)'$, the middle~$\log M+1$ bits are~$E_H(\boldd_4)'$, and the top~$M-\log ML -\log M -2$ bits are~$\boldd_4'$ (see Eq.~\eqref{Equation:informationConcise}). Now, let
\begin{alignat}{12}\label{Equation:concatenate}
\bolds_1' &= (&&\hat{\bolda}_{1},&&~\ldots~,&&\hat{\bolda}_{M},&&\hat{\boldb}_{\tau(1)},&&~\ldots~,&&\hat{\boldb}_{\tau(M)},&&\hat{\boldc}_{\rho(1)},&&~\ldots~,&&\hat{\boldc}_{\rho(M)}),	\nonumber\\
\bolds_2' &= (&&\hat{\bolda}_{\tau^{-1}(1)},&&~\ldots~,&&\hat{\bolda}_{\tau^{-1}(M)},&&\hat{\boldb}_{1},&&~\ldots~,&&\hat{\boldb}_{M},&&\hat{\boldc}_{\tau^{-1}\rho(1)},&&~\ldots~,&&\hat{\boldc}_{\tau^{-1}\rho(M)}),  \mbox{ and}  \\
\bolds_3' &= (&&\hat{\bolda}_{\rho^{-1}(1)},&&~\ldots~,&&\hat{\bolda}_{\rho^{-1}(M)},&&\hat{\boldb}_{\rho^{-1}\tau(1)},&&~\ldots~,&&\hat{\boldb}_{\rho^{-1}\tau(M)},&&\hat{\boldc}_{1},&&~\ldots~,&&\hat{\boldc}_{M}).\nonumber
\end{alignat}
The following lemma shows that at least two of the above~$\bolds_i'$ are close in Hamming distance to their encoded counterpart~$(\bolds_i,E_H(\bolds_i))$.

\begin{lemma}\label{lemma:twoAreClose}
    There exist distinct integers~$k,\ell\in [3]$ such that
    \begin{align*}
        d_H((\bolds_{k}',E_H(\bolds_k)'),(\bolds_k,E_H(\bolds_k))&\le 1\mbox{, and}\\   
        d_H((\bolds_{\ell}',E_H(\bolds_\ell)'),(\bolds_\ell,E_H(\bolds_k)))&\le 1.
    \end{align*}
\end{lemma}
\begin{proof}
    If the substitution did not occur at either of index sets $\{ 1,\ldots,L/3-1 \}$, $\{L/3+1,\ldots,2L/3-1\}$, or $\{{2L/3+1},\ldots,$ ${L-1}\}$ (which correspond to the values of the~$\bolda_i$'s, $\boldb_i$'s, and $\boldc_i$'s, respectively), then the order among the~$\bolda_i$'s, $\boldb_i$'s and~$\boldc_i$'s is maintained. That is, we have that
\begin{alignat*}{12}
\bolds_1' &= (&&{\bolda}_{1},&&~\ldots~,&&{\bolda}_{M},&&{\boldb}_{\sigma(1)},&&~\ldots~,&&{\boldb}_{\sigma(M)},&&{\boldc}_{\pi(1)},&&~\ldots~,&&{\boldc}_{\pi(M)}),	\nonumber\\
\bolds_2' &= (&&{\bolda}_{\sigma^{-1}(1)},&&~\ldots~,&&{\bolda}_{\sigma^{-1}(M)},&&{\boldb}_{1},&&~\ldots~,&&{\boldb}_{M},&&{\boldc}_{\sigma^{-1}\pi(1)},&&~\ldots~,&&{\boldc}_{\sigma^{-1}\pi(M)}),  \nonumber\\
\bolds_3' &= (&&{\bolda}_{\pi^{-1}(1)},&&~\ldots~,&&{\bolda}_{\pi^{-1}(M)},&&{\boldb}_{\pi^{-1}\sigma(1)},&&~\ldots~,&&{\boldb}_{\pi^{-1}\sigma(M)},&&{\boldc}_{1},&&~\ldots~,&&{\boldc}_{M}),
\end{alignat*}
and in this case, the claim is clear. It remains to show the other cases, and due to symmetry, assume without loss of generality that the substitution occurred in one of the~$\bolda_i$'s, i.e., in an entry which is indexed by an integer in~$\{1,\ldots,L/3-1\}$.

Let~$A\in\{0,1\}^{M\times L}$ be a matrix whose rows are the~$\boldx_i$'s, in any order. Let~$\Al$ be the result of ordering the rows of~$A$ according to the lexicographic order of their~$1,\ldots,L/3-1$ entries. Similarly, let~$\Am$ and~$\Ar$ be the results of ordering the rows of~$A$ by their $L/3+1,\ldots,2L/3-1$ and~$2L/3+1,\ldots,L-1$ entries, respectively, and let~$\Apl,\Apm$, and~$\Apr$ be defined analogously with~$\{\boldx_i'\}_{i=1}^M$ instead of~$\{\boldx_i\}_{i=1}^M$.

It is readily verified that there exist permutation matrices~$P_1$ and~$P_2$ such that~$\Am=P_1\Al$ and~$\Ar=P_2\Al$. Moreover, since $\{\boldb_i\}_{i=1}^M=\{\hat{\boldb_i}\}_{i=1}^M$, and $\{\boldc_i\}_{i=1}^M=\{\hat{\boldc}_i\}_{i=1}^M$, it follows that~$\Apm=P_1(\Al+R)$ and~$\Apr=P_2(\Al+R)$, where~$R\in\{0,1\}^{M\times L}$ is a matrix of Hamming weight~$1$; this clearly implies that~$\Apm=\Am+P_1R$ and that~$\Apr=\Ar+P_2R$. Now, notice that~$\bolds_2$ 
result from vectorizing some submatrix~$M_2$ of~$\Am$, and $\bolds_2'$ 
result from vectorizing some submatrix~$M_2'$ of~$\Am'$. Moreover, the matrices~$M_2$ and~$M_2'$ are taken from their mother matrix by omitting the same rows and columns, and both vectorizing operations consider the entries of~$M_2$ and~$M_2'$ in the same order. In addition, the redundancies~$E_H(\bolds_2)$ and~$E_H(\bolds_3)$ can be identified similarly, and have at most a single substitution with respect to the corresponding entries in the noiseless codeword. Therefore, it follows from~$\Apm=\Am+P_1R$ that~$d_H(\bolds_2',(\bolds_2,E_H(\bolds_2)))\le 1$. The claim for~$\bolds_3$ is similar.
\end{proof}

By applying a Hamming decoder on either one of the~$\bolds_i$'s, the decoder obtains possible candidates for~$\{\bolda_i\}_{i=1}^M$, $\{\boldb_i\}_{i=1}^M$, and $\{\boldc_i\}_{i=1}^M$, and by Lemma~\ref{lemma:twoAreClose}, it follows that these sets of candidates will coincide in at least two cases. Therefore, the decoder can apply a majority vote of the candidates from the decoding of each~$\bolds_i'$, and the winning values are~$\{\bolda_i\}_{i=1}^M$, $\{\boldb_i\}_{i=1}^M$, and $\{\boldc_i\}_{i=1}^M$. Having these correct values, the decoder can sort~$\{\boldx_i'\}_{i=1}^M$ according to their~$\bolda_i$ columns, and deduce the values of~$\sigma$ and~$\pi$ by observing the resulting permutation in the~$\boldb_i$ and~$\boldc_i$ columns, with respect to their lexicographic ordering. This concludes the decoding of the values~$d_1,d_3,$ and~$d_5$ of the data~$d$.

We are left to extract~$d_2,d_4,$ and~$d_6$. To this end, observe that since the correct values of~$\{\bolda_i\}_{i=1}^M$, $\{\boldb_i\}_{i=1}^M$, and $\{\boldc_i\}_{i=1}^M$ are known at this point, the decoder can extract the \textit{true} positions of~$\boldd_2,\boldd_4,$ and~$\boldd_6$, as well as their respective redundancy bits~$E_H(\boldd_2)$, $E_H(\boldd_4)$, $E_H(\boldd_6)$. Hence, the decoding algorithm is complete by applying a Hamming decoder.

We now turn to compute the redundancy of the above code~$\cC$. Note that there are two sources of redundancy---the Hamming code redundancy, which is at most~$3(\log ML +\log M + 2)$ and the fact that the sets~$\{\bolda_i\}_{i=1}^M$, $\{\boldb_i\}_{i=1}^M$, and $\{\boldc_i\}_{i=1}^M$ contain distinct strings. By a straightforward computation, for $4\le M\le 2^{L/6}$ we have


\begin{align}\label{equation:redundancysingle}
r(\cC)&=\log \binom{2^L}{M}-\log \left( \binom{2^{L/3-1}}{M}^3\cdot(M!)^2\cdot2^{3(M-\log ML-\log M-2)}\right)\nonumber\\
&=\log \prod_{i=0}^{M-1}(2^L-i) -\log  \prod_{i=0}^{M-1}(2^{L/3-1}-i)^3 - 3M+3\log ML + 3\log M +6\nonumber\\
&=\log \prod_{i=0}^{M-1}\frac{(2^L-i)}{(2^{L/3}-2i)^3}+3\log ML +3\log M+6\nonumber\\
&\le 3M\log \frac{2^{L/3}}{2^{L/3}-2M} + 3\log ML +3\log +6.\nonumber\\
&\overset{(a)}{\le} 12\log e + 3\log ML +3\log M+6
\end{align}
where inequality~$(a)$ is derived in Appendix~\ref{section:redundancy}.
 
For the case when~$M<\log ML + \log M$, we generate~$\{\bolda_i\}^M_{i=1}$,~$\{\boldb_i\}^M_{i=1}$, and~$\{\boldc_i\}^M_{i=1}$ with length~$L/3-\lceil \frac{\log ML +\log M}{M}\rceil$. As a result, we have~$\lceil \frac{\log ML +\log M}{M}\rceil$ bits~$x_{i,j}$,~$i\in\{1,\ldots,M\}$,~$j\in\{L/3-\lceil \frac{\log ML +\log M}{M}\rceil+1,\ldots,L/3\}\cup\{2L/3-\lceil \frac{\log ML +\log M}{M}\rceil+1,\ldots,2L/3\}\cup\{L-\lceil \frac{\log ML +\log M}{M}\rceil+1,\ldots,L\}$ to accommodate the information bits $\boldd_2,\boldd_4,\boldd_6$ and the redundancy bits $\{E_H(\bolds_i)\}^3_{i=1}$ and~$\{E_H(\boldd_i)\}_{i\in\{2,4,6\}}$ in each part.

\begin{remark}
    The above construction is valid whenever~$M\le 2^{L/3-1}$. However, asymptotically optimal amount of redundancy is achieved for~$M\le 2^{L/6}$. 
\end{remark}

\begin{remark}\label{remark:singlesubstitution}
In this construction, the separate storage of the Hamming code redundancies~$E_H(\boldd_2),E_H(\boldd_4)$, and~$E_H(\boldd_6)$ is not necessary. Instead, storing~$E_H(\boldd_2,\boldd_4,\boldd_6)$ is sufficient, since the true position of those can be inferred after~$\{\bolda_i\}_{i=1}^M$, $\{\boldb_i\}_{i=1}^M$, and $\{\boldc_i\}_{i=1}^M$ were successfully decoded.
This approach results in redundancy of~$3\log ML +\log 3M+O(1)$, and a similar approach can be utilized in the next section as well. 
\end{remark}
\section{Codes for Multiple Substitutions}\label{section:MultipleSub}
In this section we extend the $1$-substitution correcting code from Section~\ref{section:OneSubSimple} to multiple substitutions whenever the  number of substitutions~$K$ is at most~$L/4\log M-1/2$. In particular, we obtain the following result.
\begin{theorem}\label{theorem:multipleSubs}
	For integers~$M,L$, and~$K$ such that~$M\le 2^{\frac{L}{2(2K+1)}} $ there exists a $K$-substitution code with redundancy $$ 2K(2K+1)\log ML + 2K(2K+1)\log M+O(K).$$ 
\end{theorem}
We restrict our attention to values of~$M,L$, and~$K$ for which~$2K\log ML +2K\log M\le M$. For the remaining values, i.e., when~$2K\log ML +2K\log M> M$, a similar code can be constructed. The construction of a~$K$-substitution correcting code is similar in spirit to the single substitution case, except that we partition the strings to~$2K+1$ parts instead of~$3$. In addition, we use a Reed-Solomon code in its binary representation (see Section~\ref{section:preliminaries}) to combat~$K$-substitutions in the classic sense.
The motivation behind considering~$2K+1$ parts is that~$K$ substitutions can affect at most~$K$ of them. As a result, at least~$K+1$ parts retain their original order; and that enables a classic RS decoding algorithm to succeed. In turn, the true values of the parts are decided by a majority vote, which is applied over a set of~$2K+1$ values, $K+1$ of whom are guaranteed to be correct.

For parameters~$M,L$, and~$K$ as above, let
\begin{align*}
D=\{1,\ldots,\binom{2^{L/(2K+1)-1}}{M}^{2K+1}\cdot(M!)^{2K}\cdot2^{(2K+1)(M-2Klog ML-2K\log M)} \}
\end{align*} 
be the information set. We split a message~$d\in D$ into~$d=(d_1,\ldots,d_{4K+2})$, where~$d_1\in\{1,\ldots,\binom{2^{L/(2K+1)-1}}{M}\}$,~$d_j\in\{1,\ldots,\binom{2^{L/(2K+1)-1}}{M}M!\}$ for~$j\in\{2,\ldots,2K+1\}$, and~$d_{j}\in\{1,\ldots,2^{(2K+1)(M-2K\log ML-2K\log M)}\}$ for~$j\in\{2K+2,\ldots,4K+2\}$.
As in~\eqref{equation:combNum}, we apply~$F_{com}$ and~$F$ to obtain
\begin{align*}
    F_{com}(d_1)&=\phantom{(}\{\bolda_{1,1},\ldots,\bolda_{M,1}\},\mbox{ where }  \bolda_{i,1}\in\{0,1\}^{L/(2K+1)-1}\mbox{ for all }i, \mbox{ and }\\
    F(d_j)&=(\{\bolda_{1,j},\ldots,\bolda_{M,j}\},\pi_j)\mbox{ for all }j\in\{2,\ldots,2K+1\},\mbox{ where }\bolda_{i,j}\in\{0,1\}^{L/(2K+1)-1}\mbox{ and }\pi_j\in S_M.
\end{align*}
As usual, the sets~$\{\bolda_{i,j}\}_{i=1}^M$ are indexed lexicographically according to~$i$, i.e., $\bolda_{1,j}< \ldots < \bolda_{M,j}$ for all~$j$.
Similar to~\eqref{equation:DefXsWithabc}, let
\begin{align*}
(x_{i,(j-1)L/(2K+1)+1},\ldots,x_{i,jL/(2K+1)-1}) =\bolda_{\pi_j(i),j},~i\in[M],~j\in[2K+1].
\end{align*}
In addition, define the equivalents of~\eqref{Equation:concatenateSs} as 
\begin{alignat*}{14}
    \bolds_1 &=(&&\bolda_{1,1},            &&\ldots,&&\bolda_{M,1},~&&\bolda_{\pi_2(1),2},&&\ldots,&&\bolda_{\pi_2(M),2},&&~\ldots,~&&\bolda_{\pi_{2K+1}(1),2K+1},&&\ldots,&&\bolda_{\pi_{2K+1}(M),2K+1}),\\
    \bolds_2 &=(&&\bolda_{\pi_2^{-1}(1),1},&&\ldots,&&\bolda_{\pi_2^{-1}(M),1},~&&\bolda_{1,2},&&\ldots,~&&\bolda_{M,2},&&~\ldots,&&\bolda_{\pi_2^{-1}\pi_{2K+1}(1),2K+1},&&\ldots,&&\bolda_{\pi_2^{-1}\pi_{2K+1}(M),2K+1}),\\
    &\vdots\\
    \bolds_{2K+1} &=(&&\bolda_{\pi_{2K+1}^{-1}(1),1},&&\ldots,&&\bolda_{\pi_{2K+1}^{-1}(M),1},~&&\bolda_{\pi_{2K+1}^{-1}\pi_2(1),2},&&\ldots,~&&\bolda_{\pi_{2K+1}^{-1}\pi_2(M),2},&&~\ldots,&&\bolda_{1,2K+1},&&\ldots,&&\bolda_{M,2K+1}).
\end{alignat*}
Namely, for every~$i\in [2K+1]$, the elements~$\{\bolda_{i,j}\}_{j=1}^M$ appear in~$\bolds_i$ by their lexicographic order, and the remaining ones are sorted accordingly. 

To state the equivalent of~\eqref{Equation:informationConcise}, for a binary string~$\boldt$ let~$RS_K(\boldt)$ be the redundancy bits that result from K-substitution correcting RS encoding of~$\boldt$, in its binary representation\footnote{To avoid uninteresting technical details, it is assumed henceforth that RS encoding in its binary form is possible, i.e., that~$\log(|\boldt|)$ is an integer that divides~$\boldt$; this can always be attained by padding with zeros. Furthermore, the existence of an RS code is guaranteed, since~$q=2^{\log(|\boldt|)}$ is larger than the length of the code, which is~$|\boldt|/\log(|\boldt|)$.}. In particular, we employ an RS code which corrects~$K$ substitutions, and incurs~$2K\log(|\boldt|)$ bits of redundancy. Then, the remaining bits~$\{ x_{i,\frac{L}{2K+1}} \}_{i=1}^M$, $\{ x_{i,\frac{2L}{2K+1}} \}_{i=1}^M$, $\ldots$ , $\{ x_{i,L} \}_{i=1}^M$ are defined as follows. In this expression, notice that~$|\bolds_i|=M(L-2K-1)$ for every~$i$  and~$|\boldd_j|\le M$ for every~$j$. As a result, it follows that~$|RS_K(\boldd_j)|\le 2K\log M$ for every~$j\in\{2K+2,\ldots,4K+2\}$, and~$|RS_K(\bolds_i)|\le 2K\log ML$ for every~$i\in[2K+1]$.
\begin{align}\label{equation:defColumns}
x_{i,\frac{jL}{2K+1}} = \begin{cases}
d_{j+2K+1,i} & \mbox{if }\pi_j^{-1}(i)\le M - 2K\log M -2K\log ML\\
RS_K(\boldd_{j+2K+1})_{i-M+2K\log M+2K\log ML}, &\text{if }M-2K\log M-2K\log ML+1 \le \pi_j^{-1}(i)\le M-2K\log ML\\
RS_K(\bolds_j)_{i-M+2K\log ML}, &\text{if } M-2K\log ML+1\le \pi_j^{-1}(i)
\end{cases}.
\end{align}

To verify that the above construction provides a $K$-substitution code, observe first that~$\{ \bolda_{i,j} \}_{j=1}^M$ is a set of distinct integers for all~$i\in[2K+1]$, and hence~$d_H(\boldx_i,\boldx_j)\ge 2K+1$ for all distinct~$i$ and~$j$ in~$[M]$. Thus, a~$K$-substitution error cannot turn one~$\boldx_i$ into another, and the result is always in~$\Space$.

The decoding procedure also resembles the one in Section~\ref{section:OneSubSimple}.
Upon receiving a word~$C'=\{\boldx'_1,\ldots,\boldx'_M\}\in \cB_K(C)$ for some codeword~$C$, we define
\begin{align*}
    \hat{\bolda}_{i,j}=(x'_{\tau_j^{-1}(i),\frac{(j-1)L}{2K+1}+1},\ldots,x'_{\tau_j^{-1}(i),\frac{jL}{2K+1}-1},\ldots,) , \mbox{ for }j\in[2K+1], \mbox{ and }i\in[M]
\end{align*}
where
$\tau_j$ is the permutation by which~$\{\boldx'_i\}_{i=1}^M$ are sorted according to their $\frac{(j-1)L}{2K+1}+1,\ldots,\frac{jL}{2K+1}-1$ entries ($\tau_1$ is the identity permutation, compare with~\eqref{equation:abcHat}). In addition, sorting~$\{\boldx_i'\}_{i=1}^M$ by either one of~$\tau_j$ yields \textit{candidates} for~$\{ RS_K(\bolds_i) \}_{i=1}^{2K+1}$, for $\{\boldd_j\}_{j=2K+2}^{4K+2}$, and for~$\{ RS_K(\boldd_j) \}_{j=2K+2}^{4K+2}$. The respective~$\{\bolds_i'\}_{i=1}^{2K+1}$ are defined as
\begin{alignat*}{14}
    \bolds_1' &=(&&\hat{\bolda}_{1,1},            &&\ldots,&&\hat{\bolda}_{M,1},~&&\hat{\bolda}_{\tau_2(1),2},&&\ldots,&&\hat{\bolda}_{\tau_2(M),2},~\ldots\\
    ~ &~ &&~&&~&&~
    &&\hat{\bolda}_{\tau_{2K+1}(1),2K+1},&&\ldots,&&\hat{\bolda}_{\tau_{2K+1}(M),2K+1}),\\
    \bolds_2' &=(&&\hat{\bolda}_{\tau_2^{-1}(1),1},&&\ldots,&&\hat{\bolda}_{\tau_2^{-1}(M),1},~&&\hat{\bolda}_{1,2},&&\ldots,~&&\hat{\bolda}_{M,2},~\ldots\\
    ~ &~ &&~&&~&&~
    &&\hat{\bolda}_{\tau_2^{-1}\tau_{2K+1}(1),2K+1},&&\ldots,&&\hat{\bolda}_{\tau_2^{-1}\tau_{2K+1}(M),2K+1}),\\
    &\vdots\\
    \bolds_{2K+1}' &=(&&\hat{\bolda}_{\tau_{2K+1}^{-1}(1),1},&&\ldots,&&\hat{\bolda}_{\tau_{2K+1}^{-1}(M),1},~&&\hat{\bolda}_{\tau_{2K+1}^{-1}\tau_2(1),2},&&\ldots,~&&\hat{\bolda}_{\tau_{2K+1}^{-1}\tau_2(M),2},~\ldots\\
    ~ &~ &&~&&~&&~
    &&\hat{\bolda}_{1,2K+1},&&\ldots,&&\hat{\bolda}_{M,2K+1}).\\
\end{alignat*}
\begin{lemma}\label{lemma:MajWorksMultiple}
    There exist~$K+1$ distinct integers~$\ell_1,\ldots,\ell_{K+1}$ such that~$d_H((\bolds'_{\ell_j},RS_K(\bolds_{\ell_j})'),(\bolds_{\ell_j},RS_K(\bolds_{\ell_j})))\le K$ for every~$j\in [K+1]$.
\end{lemma}
\begin{proof}
Analogous to the proof of Lemma~\ref{lemma:twoAreClose}. See Appendix~\ref{section:omitted Proofs} for additional details.
\end{proof}

By applying an RS decoding algorithm on each of~$\{\bolds_i'\}_{i=1}^{2K+1}$ we obtain candidates for the true values of~$\{\bolda_{i,j}\}_{j=1}^M$ for every~$i\in[2K+1]$. According to Lemma~\ref{lemma:MajWorksMultiple}, at least~$K+1$ of these candidate coincide, and hence the true value of $\{\bolda_{i,j}\}_{j=1}^M$ can be deduced by a majority vote. Once these true values are known, the decoder can sort~$\{\boldx_i'\}_{i=1}^M$ by its~$\bolda_{1,j}$ entries (i.e., the entries indexed by $1,\ldots,\frac{L}{2K+1}-1$), and deduce the values of each~$\pi_t$, $t\in\{2,\ldots,2K+1\}$ according to the resulting permutation of~$\{\bolda_{t,\ell}\}_{\ell=1}^M$ in comparison to their lexicographic one. Having all the permutations~$\{\pi_j\}_{j=2}^{2K+1}$, the decoder can extract the true positions of~$\{\boldd_j\}_{j=2K+2}^{4K+2}$ and $\{RS_K(\boldd_j)\}_{j=2K+2}^{4K+2}$, and apply an RS decoder to correct any substitutions that might have occurred.

\begin{remark}
    Notice that the above RS code in its binary representation consists of binary substrings that represent elements in a larger field. As a result, this code is capable of correcting \emph{any} set of substitutions that are confined to at most~$K$ of these substrings. Therefore, our code can correct more than~$K$ substitutions in many cases.
\end{remark}


For~$4\le M\le 2^{L/2(2K+1)}$, the total redundancy of the above construction~$\cC$ is given by
\begin{align}\label{equation:redundancymultiple}
r(\cC)&=\log \binom{2^L}{M}-\log \binom{2^{L/(2K+1)-1}}{M}^{2K+1}M!^{2K}2^{(2K+1)(M-2K\log ML-2K\log M)}\nonumber\\
&\overset{(b)}{\le} (2K+1) \log e +2K(2K+1)\log ML +2K(2K+1)\log M.
\end{align}
where the proof of inequality~$(b)$ is given in Appendix~\ref{section:redundancy}.


\begin{remark}
As mentioned in Remark~\ref{remark:singlesubstitution}, storing~$RS_K(\boldd_j)$ separately in each part~$j\in\{2K+2,\ldots,4K+2\}$ is not necessary. Instead, we    store~$RS_K(\boldd_{2K+2},\ldots,\boldd_{4K+2})$ in a single part~$j=2K+1$, since the position of the binary strings~$\boldd_{j}$ for~$j\in\{2K+2,\ldots,4K+2\}$ and the redundancy~$RS_K(\boldd_{2K+2},\ldots,\boldd_{4K+2})$ can be identified once~$\{\bolda_{i,j}\}_{i\le M,j\le 2K+1}$ are determined. The redundancy of the resulting code is~$2K(2K+1)\log ML + 2K\log (2K+1)M$.
\end{remark}
For the case when~$M<2K\log ML + 2K\log M$, we generate sequences~$\bolda_{i,j}$,~$i\in\{1,\ldots,M\}$,~$j\in\{1,\ldots,2K+1\}$ with length~$L/(2K+1)-\lceil \frac{2K\log ML +2K\log M}{M}\rceil$. Then, the length~$\lceil \frac{2K\log ML +2K\log M}{M}\rceil$ sequences~$x_{i,j}$,~$i\in\{1,\ldots,M\}$,~$j\in\cup^{2K+1}_{l=1}\{(l-1)L/(2K+1)-\lceil \frac{2K\log ML +2K\log M}{M}\rceil+1,\ldots,lL/(2K+1)\}$ are used to accommodate the information bits $\{\boldd_j\}^{4K+2}_{j=2K+2}$ and the redundancy bits $\{RS_K(\bolds_i)\}^{2K+1}_{i=1}$ and~$\{RS_K(\boldd_j)\}^{4K+2}_{j=2K+2}$ in each part.


\section{Codes with order-wise optimal redundancy}\label{section:optimal}
In this section we briefly describe how to construct~$K$-substitution correcting codes whose redundancy is order-wise optimal, in a sense that will be clear shortly. The code construction applies whenever~$K$ is at most $O(\min\{L^{1/3},L/\log M\})$. 
\begin{theorem}\label{theorem:multipleoptimal}
	For integers~$M,L$, and~$K$, let~$L'=3\log M  + 4K^2+1$. If~$L'+4KL'+2K\log (4KL')\le L$, then there exists an explicit $K$-substitution code with redundancy $ 2K\log ML + (12K+2)\log M+O(K^3)+O(K\log\log ML)$ 
\end{theorem}
As in Section~\ref{section:OneSubSimple} and Section~\ref{section:MultipleSub}, we use the information bits themselves for the purpose of indexing.
Specifically,
we encode information in the first~$L'$ bits~$(x_{i,1},x_{i,2},\ldots,x_{i,L'})$ in each sequence~$\boldx_i$ and then sort the sequences~$\{\boldx_i\}^M_{i=1}$ according to the lexicographic order~$\pi$ of $(x_{i,1},x_{i,2},\ldots,x_{i,L'})$, such that~$(x_{\pi (i),1},x_{\pi(i),2},\ldots,x_{\pi(i),L'})<(x_{\pi (j),1},x_{\pi(j),2},\ldots,x_{\pi(j),L'})$ for~$i<j$. Then, we protect the sequences $\{\boldx_i\}^M_{i=1}$ in the same manner as if they are ordered, i.e.,  by concatenating them and applying a Reed-Solomon encoder. 

An issue that must be addressed is how to protect
the ordering~$\pi$ from being affected by substitution errors. This is done in two steps: (1) Using additional redundancy to protect the ordering sequence set~$\{(x_{i,1},x_{i,2},\ldots,x_{i,L'})\}^M_{i=1}$; and (2) Constructing~$\{(x_{i,1},x_{i,2},\ldots,x_{i,L'})\}^M_{i=1}$ such that the Hamming distance between any two distinct sequences~$(x_{i,1},x_{i,2},\ldots,x_{i,L'})$ and~$(x_{j,1},x_{j,2},\ldots,x_{j,L'})$ is at least ~$2K+1$. In this way, the bits~$(x_{i,1},x_{i,2},\ldots,x_{i,L'})$ in sequence~$\boldx_i$ can be recovered from their erroneous version~$(x'_{i,1},x'_{i,2},\ldots,x'_{i,L'})$, which is within Hamming distance~$K$ from~$(x_{i,1},x_{i,2},\ldots,x_{i,L'})$. The details of the encoding and decoding are as follows.

For an integer~$n$, let~$\1_n$ be the vector of~$n$ ones.
Let~$\mathcal{S}$ be the ensemble of all codes of length~$L'$, cardinality~$M$, and minimum Hamming distance at least~$2K+1$, which contain~$\1_{L'}$, that is,
\begin{align*}
    \cS\triangleq\left\{\{\bolda_1,\ldots,\bolda_M\}\in\binom{\{0,1\}^{L'}}{M}\Big\vert\bolda_1= \1_{L'}\mbox{ and }
d_H(\bolda_i,\bolda_j)\ge 2K+1\mbox{ for every distinct }i,j\in[M]\right\}.
\end{align*}
Now we show that
\begin{align}\label{equation:numberofcodes}
    |\mathcal{S}|\ge \frac{\prod^{M}_{i=2}[2^{L'}-(i-1)Q]}{(M-1)!},
\end{align}
where $Q=\sum^{2K}_{i=0}\binom{L'}{i}$ is the size of a Hamming ball of radius~$2K$ centered at a vector in~$\{0,1\}^{L'}$. 
For $$\cS_T=\left\{(\bolda_1,\ldots,\bolda_{M}):\bolda_1= \1_{L'}\mbox{ and }d_H(\bolda_i,\bolda_j)\ge 2K+1 \mbox{ for distinct }i,j\in[M]\right\},$$ it is shown that~$|\cS_T|\ge \prod^{M}_{i=2}[2^{L'}-(i-1)Q]$. The idea is to let~$\bolda_1=\1_{L'}$ and then select~$\bolda_2,\ldots,\bolda_{M}$ sequentially while keeping the mutual Hamming distance among~$\bolda_1,\ldots,\bolda_{i}$  at least~$2K+1$ for~$i\in\{2,\ldots,M\}$. 
Notice that for any sequence~$\bolda\in\{0,1\}^{L'}$, there are at most~$Q$ sequences that are within Hamming distance~$2K$ of~$\bolda$. Hence,  
given~$\bolda_1=\1_{L'},\ldots,\bolda_{i-1}$,~$i\in\{2,\ldots,M\}$ such that the mutual Hamming distance among~$\bolda_1,\ldots,\bolda_{i-1}$ is at least~$2K+1$, there are at least~$2^{L'}-iQ$ choices of~$\bolda_i\in\{0,1\}^{L'}$,~$i\in\{2,\ldots,M\}$ such that the Hamming distance between~$\bolda_i$ and each one of~$\bolda_1,\ldots,\bolda_{i-1}$ is at least~$2K+1$. These choices of~$\bolda_i\in\{0,1\}^{L'}$ keep the mutual Hamming distance among~$\bolda_1,\ldots,\bolda_{i}$ at least~$2K+1$.
Therefore, 
it follows that~$|\cS_T|\ge \prod^{M}_{i=2}(2^{L'}-(i-1)Q)$. Since there are~$(M-1)!$ tuples in~$\cS_T$ that correspond to the same set~$\{\1_{L'},\bolda_2,\ldots,\bolda_{M}\}$ in~$\mathcal{S}$, we have that equation~\eqref{equation:numberofcodes} holds. 

According to~\eqref{equation:numberofcodes}, there exists an injective mapping~$F_S:\left[\lceil\frac{\prod^{M-1}_{i=1}(2^{L'}-iQ)}{(M-1)!}\rceil\right]\rightarrow \binom{\{0,1\}^{L'}}{M}$ that maps an integer~$i\in\{1,\ldots,\lceil\frac{\prod^{M-1}_{i=1}(2^{L'}-iN)}{(M-1)!}\rceil\}$ 
to a code~$S\in \mathcal{S}$. The mapping~$F_S$ is invertible and can be computed in~$O(2^{ML'})$ time using brute force. We note that there is a greedy algorithm implementing the mapping~$F_S$ and the corresponding inverse mapping~$F^{-1}_S$ in~$Poly(M,L,k)$ time. We defer this algorithm and the corresponding analysis to a future version of this paper.
For~$S\in \mathcal{S}$, define the characteristic vector~$\1(S)\in\{0,1\}^{2^{L'}}$ of~$S$ by
\begin{align*}
    		\1 (S)_i=\begin{cases}
			1 & \mbox{if the binary presentation of~$i$ is in~$S$}\\
			0 &\mbox{else} 
		\end{cases}.
\end{align*}
Notice that the Hamming weight of~$\1(S)$ is~$M$ for every~$S\in\cS$. 
Intuitively, we use the choice of a code~$S\in\cS$ to store information, and the lexicographic order of the strings in the chosen~$S$ to order the strings in our codewords, and the details are as follows.

Consider the data~$\boldd\in D$ to be encoded as a tuple~$\boldd=(d_1,\boldd_2)$, where~$d_1\in \{1,\ldots, \lceil\frac{\prod^{M-1}_{i=1}(2^{L'}-iQ)}{(M-1)!}\rceil\}$ and 
$$\boldd_2\in \{0,1\}^{M(L-L')-4KL'-2K\lceil \log (4KL')\rceil-2K\lceil\log ML\rceil}.$$
Given~$(d_1,\boldd_2)$, the codeword~$\{\boldx_i\}^M_{i=1}$ is generated by the following procedure.

\textbf{Encoding:} 
\begin{itemize}
    \item [\textbf{(1)}] Let~$F_S(d_1)=\{\bolda_1,\ldots,\bolda_M\}\in\mathcal{S}$ such that~$\bolda_1=\1_{L'}$ and the $\bolda_i$'s are sorted in a descending lexicographic order.
    Let $(x_{i,1},\ldots,x_{i,L'})=\bolda_i$, for~$i\in[M]$.
    \item [\textbf{(2)}] Let $(x_{1,L'+1},\ldots,x_{1,L'+4KL'})=RS_{2K}(\1 (\{\bolda_1,\ldots,\bolda_M\}))$ (see the paragraph before~\eqref{equation:defColumns} for the definition of~$RS_K(\boldt)$) and
    $$(x_{1,L'+4KL'+1},\ldots,x_{1,L'+4KL'+2K\lceil \log (4KL')\rceil})=RS_K(RS_{2K}(\1 (\{\bolda_1,\ldots,\bolda_M\}))).$$
    \item [\textbf{(3)}] Place the information bits of~$\boldd_2$ in bits
    \begin{align*}
        &(x_{1,L'+4KL'+2K\lceil \log (4KL')\rceil+1},\ldots,x_{1,L}),\\ &(x_{M,L'+1},\ldots,x_{M,L-2K\lceil\log ML\rceil})\mbox{; and}\\ &(x_{i,L'+1},\ldots,x_{i,L})\mbox{ for }i\in\{2,\ldots,M-1\}.
    \end{align*}
    \item [\textbf{(4)}] Define
    \begin{align*}
        \bolds = (\boldx_1,\ldots,\boldx_{M-1},(x_{M,1},\ldots,x_{M,L-2K\lceil\log ML\rceil}))
    \end{align*}
    and let $(x_{M,L-2K\lceil\log ML\rceil+1},\ldots,x_{M,L})=RS_K(\bolds)$.
\end{itemize} 
Upon receiving the erroneous version\footnote{Since the sequences~$\{\boldx_i\}^M_{i=1}$ have distance at least~$2K+1$ with each other, the sequences~$\{\boldx'_i\}^M_{i=1}$ are different.}~$(\boldx'_1,\ldots,\boldx'_M)$, the decoding procedure is as follows. 

\textbf{Decoding:} 
\begin{itemize}
    \item [\textbf{(1)}] Find the unique sequence~$\boldx'_{i_0}$ such that~$(x'_{i_0,1},\ldots,$ $x'_{i_0,L'})$ has at least~$L'-K$  many~$1$-entries. By the definition of~$\{(x_{i,1},\ldots,x_{i,L'})\}_{i=1}^M$, we have that~$\boldx'_{i_0}$ is an erroneous copy of~$\boldx_1$. Then, use a Reed-Solomon decoder to decode bits~$(x_{1,L'+1},\ldots, $ $x_{1,L'+4KL'})$ from $$(x'_{i_0,L'+1},\ldots,x'_{i_0,L'+4KL'+2K\lceil \log (4KL')\rceil}).$$Note that~$(x'_{i_0,L'+1},$ $\ldots,$ $x'_{i_0,L'+4KL'+2K\lceil \log (4KL')\rceil})$ is an erroneous copy of~$(x_{1,L'+1},\ldots,x_{1,L'+2KL'+2K\lceil \log (2KL')\rceil})$, which by definition is a codeword in a Reed-Solomon code.
    \item [\textbf{(2)}] Use a Reed-Solomon decoder and the Reed-Solomon redundancy~$(x_{1,L'+1},\ldots, x_{1,L'+4KL'})$ to recover the vector $\1 (\{(x_{i,1},$ $\ldots,x_{i,L'})\}^M_{i=1})$ and then the set $\{(x_{i,1},\ldots,x_{i,L'})\}^M_{i=1}$. 
    Since $\1 (\{(x_{i,1},\ldots,x_{i,L'})\}^M_{i=1})$ is within Hamming distance~$2K$ from~$\1 (\{(x'_{i,1},\ldots,x'_{i,L'})\}^M_{i=1})$, the former
    can be recovered given its Reed-Solomon redundancy~$(x_{1,L'+1},\ldots,$ $ x_{1,L'+4KL'})$. 
    \item [\textbf{(3)}] For each~$i\in[M]$, find the unique~$\pi(i)\in [M]$ such that~$d_H((x'_{\pi(i),1},\ldots,x'_{\pi(i),L'}),(x_{i,1},\ldots,x_{i,L'}))\le K$  (note that~$\pi(i_0)=1$), and conclude that~$\boldx'_i$ is an erroneous copy of~$\boldx_{\pi(i)}$.
    \item [\textbf{(4)}] With the order~$\pi$ recovered, concatenate~$(\boldx'_{\pi^{-1}(1)},\ldots,\boldx'_{\pi^{-1}(M)})$. We have that~$(\boldx'_{\pi^{-1}(1)},\ldots,\boldx'_{\pi^{-1}(M)})$ is an erroneous copy of $(\boldx_{1},\ldots,\boldx_{M})$, which by definition is a codeword in a Reed-Solomon code. Therefore,~$(\boldx_{1},\ldots,\boldx_{M})$ can be recovered from $(\boldx'_{\pi^{-1}(1)},\ldots,\boldx'_{\pi^{-1}(M)})$.   
\end{itemize}
The redundancy of the code is 
\color{black}
\begin{align}\label{equation:inequality2}
    r(\mathcal{C})
    =&\log \binom{2^L}{M}-\log \lceil\frac{\prod^{M-1}_{i=1}(2^{L'}-iQ)}{(M-1)!}\rceil\nonumber\\
    \le &2K\log ML + (12K+2)\log M+O(K^3)+O(K\log\log ML),
\end{align}
which will be proved in Appendix~\ref{section:proofofinequality}.
\begin{remark}
Note that the the proof of~\eqref{equation:numberofcodes} indicates an algorithm for computing mapping~$F_S(i)$ with complexity exponential in~$L'$ and~$M$. A $poly(M,L)$ complexity algorithm that computes~$F_S(i)$ will be given in future versions of this paper.
\end{remark}

\section{Conclusions and Future Work}\label{section:FutureWork}
Motivated by novel applications in coding for DNA storage, this paper presented a channel model in which the data is sent as a set of unordered strings, that are distorted by substitutions. Respective sphere packing arguments were applied in order to establish an existence result of codes with low redundancy for this channel, and a corresponding lower bound on the redundancy for~$K=1$ was given by using Fourier analysis. 
For~$K=1$, a code construction was given which asymptotically achieves the lower bound. For larger values of~$K$, a code construction whose redundancy is asymptotically~$K$ times the aforementioned upper bound was given; closing this gap is an interesting open problem. Furthermore, it is intriguing to find a lower bound on the redundancy for larger values of~$K$ as well.


\bibliographystyle{IEEEtran}

\appendices

\section{Proof of Lemma \ref{lemma:MajWorksMultiple}}\label{section:omitted Proofs}
\begin{proof}
    (of Lemma~\ref{lemma:MajWorksMultiple}) Similarly to the proof of~Lemma~\ref{lemma:twoAreClose}, we consider a matrix
~$A\in\{0,1\}^{M\times L}$  whose rows are the~$\boldx_i$'s, in any order. Let~$A_j$ be the result of ordering the rows of~$A$ according to the lexicographic order of their~$(j-1)L/(2K+1)+1,\ldots,jL/(2K+1)-1$ bits for~$j\in [2K+1]$.  The matrices~$A'_{j}$ for ~$j\in [2K+1]$ can be defined analogously with~$\{\boldx_i'\}_{i=1}^M$ instead of~$\{\boldx_i\}_{i=1}^M$.

It is readily verified that there exist ~$2K+1$ permutation matrices~$P_j$ such that~$A_j=P_j A$ (Here~$P_1$ is the identity matrix). Moreover, since~$K$ substitution spoils at most~$K$ parts, there exist at least~$j_l\in[2K+1],l\in[K+1]$ such that
$\{\bolda_{i,j_l}\}_{i=1}^M=\{\hat{\bolda_{i,j_l}}\}_{i=1}^M$, for~$l\in[K+1]$, it follows that~$A'_{j_l}=P_{j_l}(A+R)$ for~$l\in[K+1]$, where~$R\in\{0,1\}^{M\times L}$ is a matrix of Hamming weight at most~$K$; this clearly implies that~$A'_{j_l}=A_{j_l}+P_{j_l}R$ for~$l\in[K+1]$. Since~$\bolds_{j_{l}}$ 
results from vectorizing some submatrix~$M_{l}$ of~$A_{j_l}$, and $\bolds_{j_l}'$ 
results from vectorizing some submatrix~$M_l'$ of~$A_{j_l}'$.  Moreover, the matrices~$M_l$ and~$M_l'$ are taken from their mother matrix by omitting the same rows and columns, and both vectorizing operations consider the entries of~$M_l$ and~$M_l'$ in the same order.
In addition, the redundancies~$E_H(\bolds_{j_l})$ for~$l\in[K+1]$ can be identified similarly, and have at most~$K$ substitution with respect to the corresponding entries in the noiseless codeword.
Therefore, it follows from~$A_{j_l}=A_{j_l}+P_1R$ that~$d_H((\bolds_{j_l}',,E_H(\bolds_{j_l})),(\bolds_{j_l},E_H(\bolds_{j_l})))\le K$. 

\end{proof}
\section{Proof of Redundancy Bounds}\label{section:redundancy}
\emph{Proof of~$(a)$ in \eqref{equation:redundancysingle}:}
\begin{align*}
r(\cC)&\le 3\log (1+
\frac{2M}{2^{L/3}-2M})^M	+ 3\log ML +3\log M+6\\
&\le 3\log (1+\frac{4}{M})^M+ 3\log ML +3\log M+6\\
&=12\log ((1+\frac{4}{M})^{M/4})+ 3\log ML +3\log M+6 \\
&\le 12\log e + 3\log ML +3\log M+6.
\end{align*}
\emph{Proof of~$(b)$ in \eqref{equation:redundancymultiple}:}
\begin{align*}
r(\cC)
&=\log \prod_{i=0}^{M-1}(2^L-i) -\log  \prod_{i=0}^{M-1}(2^{L/(2K+1)-1}-i)^{2K+1} - \log 2^{(2K+1)M} +2K(2K+1)\log ML +2K(2K+1)\log  M\\
&=\log \prod_{i=0}^{M-1}\frac{(2^L-i)}{(2^{L/(2K+1)}-2i)^{2K+1}}+2K(2K+1)\log ML +2K(2K+1)\log  M\\
&\le (2K+1)M\log \frac{2^{L/(2K+1)}}{2^{L/(2K+1)}-2M} +2K(2K+1)\log ML +2K(2K+1)\log M\\
&\le (2K+1)\log (1+
\frac{2M}{2^{L/(2K+1)}-2M})^M	+2K(2K+1)\log ML +2K(2K+1)\log M\\
&\le (2K+1)\log (1+\frac{4}{M})^M+2K(2K+1)\log ML +2K(2K+1)\log M\\
&=(2K+1)\log ((1+\frac{4}{M})^{M/4})+2K(2K+1)\log ML +2K(2K+1)\log M \\
&\le (2K+1) \log e +2K(2K+1)\log ML +2K(2K+1)\log M.
\end{align*}

\section{
Improved Codes for a Single Substitution}\label{section:OneSubComplicated}
We briefly present an improved construction of a single substitution code, which achives~$2\log ML + \log 2M +O(1)$ redundancy.
\begin{theorem}
	Let~$M$ and~$L$ be numbers that satisfy~$M\le  2^{L/4}$. Then there exists a single substitution correcting code with redundancy~$2\log ML + \log 2M+O(1)$.
\end{theorem}
The construction is based on the single substitution code as shown in Section~\ref{section:OneSubSimple}. The difference is that instead of using three parts and the majority rule, it suffices to use two parts (two halfs) and an extra bit to indicate which part has the correct order. To compute this bit, let
\begin{align*}
	\boldx_{\oplus}=\bigoplus^{M}_{i=1}\boldx_i
\end{align*} 
be the bitwise XOR of all strings~$\boldx_i$ and~$\bolde\in\{0,1\}^L$ be a vector of~$L/2$ zeros followed by~$L/2$ ones. We use the bit $
b_e = \bolde \cdot \boldx_{\oplus} \bmod 2$ to indicate in which part the substitution error occurs. If a substitution error happens at the first half~$(x^1_i,\ldots,x^{L/2}_i)$, the bit~$b_e$ does not change. Otherwise the bit~$b_e$ is flipped. 
Moreover, as mentioned in Remark~\ref{remark:singlesubstitution}, we store the redundancy of all the binary strings in a single part, instead of storing the redundancy separately for each binary string in each part.
The data to encode is regarded as~$d=(d_1,d_2,d_3,d_4)$, where~$d_1\in \{1,\ldots,{2^{L/2-1}\choose M}\}$, $d_2\in \{1,\ldots,{2^{L/2-1}\choose M}\cdot M!\}$,~$d_3\in \{1,\ldots,2^{M-\log ML-1}\}$
and~$d_4\in\{1,\ldots,2^{M-\log ML-\log 2M-2}\}$.
That is, $d_1$ represents a set of~$M$ strings of length~$L/2-1$, $d_2$ represents a set of~$M$ strings of length~$L/2-1$ and a permutation~$\pi$. Let~$\boldd_3\in\{0,1\}^{M-\log ML-1},\boldd_4\in \{ 0,1\}^{M-\log ML-\log 2M-2 }$ be the binary strings corresponds to~$d_3$ and~$d_4$ respectively.

We now address the problem of inserting the bit~$b_e$ into the codeword. We consider the four bits~$x_{i_1,L/2}$,~$x_{i_2,L/2}$,~$x_{i_3,L}$, and~$x_{i_4,L}$, where~$i_1$ and~$i_2$ are the indices of the two largest strings among~$\{ \bolda_i \}_{i=1}^M$ in lexicographic order,
and $i_3$ and~$i_4$ are the indices of the two largest strings among~$\{\boldb_i\}_{i=1}^M$ in lexicographic order. Then, we compute~$b_e$ and set
\begin{align*}
    x_{i_1,L/2}=x_{i_2,L/2}=x_{i_3,L}=x_{i_4,L}=b_e.
\end{align*}
Note that after a single substitution, at most one of $i_1$,~$i_2$,~$i_3$, and~$i_4$ will not be among the indices of the largest two strings in their corresponding part. Hence, upon receiving a word~$C'=\{\boldx'_1,\ldots,\boldx'_M\}\in \cB_1(C)$ for some codeword~$C$, we find the two largest strings among~$\{\bolda_i\}_{i=1}^M$ and the two largest strings among~$\{\boldb_i\}_{i=1}^M$,
and use majority to determine the bit~$b_e$. The rest of the encoding and decoding procedures are similar to the corresponding ones in Section~\ref{section:OneSubSimple}. We define~$\bolds_1$ and~$\bolds_2$ to the two possible concatenations of~$\{\bolda_i\}_{i=1}^M$ and~$\{\boldb_i\}_{i=1}^M$,
\begin{alignat*}{11}
\bolds_1 &= (&&\bolda_{1},&&~\ldots~,&&\bolda_{M},&&\boldb_{\pi(1)},&&~\ldots~,&&\boldb_{\pi(M)})\\
\bolds_2 &= (&&\bolda_{\pi^{-1}(1)},&&~\ldots~,&&\bolda_{\pi^{-1}(M)},&&\boldb_{1},&&~\ldots~,&&\boldb_{M}).
\end{alignat*}
We compute their Hamming redundancies and place them in columns~$L/2$ and~$L$, alongside the strings~$d_3,d_4$ and their Hamming redundancy~$E_H(\boldd_3,\boldd_4)$ in column~$L$, similar to~\eqref{Equation:informationConcise}. 

In order to decode, we compute the value of~$b_e$ by a majority vote, which locates the substitution, and consequently, we find~$\pi$ by ordering~$\{\boldx_i'\}_{i=1}^M$ according to the error-free part. Knowing~$\pi$, we extract the~$d_i$'s and their redundancy~$E_H(\boldd_3,\boldd_4)$, and complete the decoding procedure by applying a Hamming decoder. The resulting redundancy is~$2\log ML+\log 2M+ 3$.

\section{Proof of~$(a)$ in Eq. \eqref{equation:rk}}\label{section:monotonicity}
Note that~$P\le T$, it suffices to show that the
function~$g(P)\triangleq  ((T+P)/P)^P= (1+T/P)^P$ is increasing in~$P$ for~$P>0$.
We now show that the derivative
$\partial g(P)/\partial P= (1+T/P)^P(\ln (1+T/P)-T/(T+P))$ is greater than~$0$ for~$P>0$. It is left to show that \begin{align}\label{equation:derivative}
   \ln(1+T/P)>T/(T+P)
\end{align}
Let~$v=T/(T+P)$, then Eq.~\eqref{equation:derivative} is equivalent to
\begin{align}\label{equation:vinequality}
    1/(1-v)>e^v
\end{align}
for some~$0<v<1$.
The inequality~\eqref{equation:vinequality} holds since~$1/(1-v)=1+\sum^\infty_{i=1}v^i$ and~$e^v=1+\sum^\infty_{i=1}v^i/i!$ for ~$0<v<1$.

\section{Proof of~Eq. \eqref{equation:inequality2}}\label{section:proofofinequality}
\begin{align*}
    r(\mathcal{C})
    =&\log \binom{2^L}{M}-\log \lceil\frac{\prod^{M-1}_{i=1}(2^{L'}-iQ)}{(M-1)!}\rceil\\
    &- [M(L-L')-4KL'-2K\lceil \log (4KL')\rceil-2K\lceil\log ML\rceil]\\
    \le &\log \frac{2^{LM}}{M!} - \log \frac{(2^{L'}-MQ)^{M-1}}{(M-1)!}\\
    &- [M(L-L')-4KL'-2K (\log (4KL')+1)-2K(\log ML+1)]\\
    = &ML'- \log (2^{L'}-MQ)^{M-1} +4KL'+2K \log (4KL')\\
    &+2K\log ML + 4K-\log M\\
    =& \log \frac{2^{L'(M-1)}}{(2^{L'}-MQ)^{M-1}} + L'+4KL'+2K \log (4KL')\\
    &+2K\log ML + 4K-\log M\\
    =& \frac{(M-1)MQ}{2^{L'}-MQ}\log (1+\frac{MQ}{2^{L'}-MQ})^{\frac{2^{L'}-MQ}{MQ}}+ L'+4KL'+2K \log (4KL')\\
    &+2K\log ML + 4K-\log M\\
    \overset{(a)}{\le} & \log e + L'+4KL'+2K \log (4KL')+2K\log ML + 1+4K-\log M\\
    =&2K\log ML + (12K+2)\log M+O(K^3)+O(K\log\log ML)
\end{align*}
where~$(a)$ follows from the following inequality
\begin{align}\label{equation:inequality22}
    M^2(3\log M +4K^2+1)^{2K}\le 2^{3\log M+4K^2+1},  
\end{align}
which is proved as follows.

Rewrite Eq.~\eqref{equation:inequality22} as
\begin{align}\label{equation:rewrite}
    (3\log M +4K^2+1)^{2K}\le 2^{\log M+4K^2+1}.
\end{align}
Define functions~$g(y,K)=\ln (3y+4K^2+1)^{2K}$ and~$h(y,K)=\ln 2^{y+4K^2+1}$. Then we have that
\begin{align*}
    \partial h(y,K)/\partial y
    -\partial g(y,K)/\partial y
    =\ln 2 - 6K/(3y+4K^2+1),
\end{align*}
which is positive for~$y\ge 1$ and~$K\ge 2$. Therefore, for~$k\ge 2$ and~$y\ge 1$, we have that
\begin{align*}
    h(y,K)-g(y,K)\ge h(1,K)-g(1,K).
\end{align*}
Furthermore,
\begin{align*}
    \partial h(1,K)/\partial K
    -\partial g(1,K)/\partial K
    =&(8\ln 2) K -2\ln (4K^2+4)-   16K^2/(4K^2+4)\\
    >&(8\ln 2) K-2\ln (5K^2)-4\\
    =&4(K-1-\ln K)+(8\ln 2-4) K-2\ln 5\\
    \overset{(a)}{\ge} &(8\ln 2-4) K-2\ln 5,
\end{align*}
where~$(a)$ follows since~$K=e^{\ln K}\ge 1+\ln K$. Since~$(8\ln 2-4) K-2\ln 5$ is positive for~$K\ge 3$, we have that~$h(1,K)/\partial K
    >\partial g(1,K)/\partial K$ for~$K\ge 3$.
It then follows that~$h(1,K)-g(1,K)\ge\min\{h(1,2)-g(1,2),h(1,3)-g(1,3)\}>0$ for~$K\ge 2$. Hence~$h(y,K)>g(y,K)$ for~$y\ge 1$ and~$K\ge 2$, which implies that Eq.~\eqref{equation:rewrite} holds when~$M\ge 2$ and~$K\ge 2$. 

Next we show that Eq.~\eqref{equation:rewrite} holds when~$M=1$ or~$K=1$.
When~$M=1$, we have that~$\log M=0$ and that
\begin{align*}
    \partial h(0,K)/\partial K
    -\partial g(0,K)/\partial K
    =&(8\ln 2) K -2\ln (4K^2+1)-   16K^2/(4K^2+1)\\
    >&(8\ln 2) K-2\ln (5K^2)-4\\
    =&4(K-1-\ln K)+(8\ln 2-4) K-2\ln 5\\
    \ge& (8\ln 2-4) K-2\ln 5,
\end{align*}
which is positive when~$K\ge 3$. Therefore, we have that~
$h(0,K)-g(0,K)\ge \min\{h(0,1)-g(0,1),h(0,2)-g(0,2),h(0,3)-g(0,3)\}>0$. Hence Eq.\eqref{equation:rewrite} holds when~$M=1$.

When~$K=1$ we have that
\begin{align*}
    2^{\log M+4K^2+1}
    =&32(1++\sum^\infty_{i=1}\log^i M/i!)\\
    \ge& 32(1+\log M+\log^2M/2)\\
    \ge &(3\log M+5)^2\\
    =&(3\log M+4K^2+1)^{2K}.
\end{align*}
Hence, Eq.~\eqref{equation:rewrite} and Eq.~\eqref{equation:inequality22} holds. We now finish the proof of Eq.~\eqref{equation:inequality2}.

\end{document}